\documentclass[sigconf]{acmart}

\newtheorem{thm}{Theorem}

\newtheorem{obs}{Observation}

\AtBeginDocument{%
  \providecommand\BibTeX{{%
    \normalfont B\kern-0.5em{\scshape i\kern-0.25em b}\kern-0.8em\TeX}}}

\setcopyright{acmcopyright}
\copyrightyear{2018}

\acmConference[KDD 2020]{The fifth international workshop on fashion and KDD in conjunction with 
The 26th ACM SIGKDD Conference on Knowledge Discovery and Data Mining}{August 23 - 27, 2020}{San Diego, CA}

\begin{document}

\title{Breaking Moravec's Paradox: Visual-Based Distribution in Smart Fashion Retail}


\author{Shin Woong Sung}
\authornotemark[2]
\authornote{Shin Woong Sung is now in Samsung Electronics}
\email{sw.sung@kaist.ac.kr}
\author{Hyunsuk Baek}
\authornote{Both authors contributed equally to this research}
\email{hisuk31@kaist.ac.kr}
\affiliation{%
  \institution{Korea Advanced Institute of Science and Technology}
  \city{Daejeon}
  \country{Republic of Korea}
}

\author{Hyeonjun Sim}
\email{flhy5836@kaist.ac.kr}
\affiliation{%
  \institution{Korea Advanced Institute of Science and Technology}
  \city{Daejeon}
  \country{Republic of Korea}
}

\author{Eun Hie Kim}
\affiliation{%
 \institution{KOLON Corporation}
 \country{Republic of Korea}}

\author{Hyunwoo Hwangbo}
\affiliation{%
  \institution{KOLON Corporation}
  \country{Republic of Korea}}

\author{Young Jae Jang}
\authornote{Corresponding author}
\email{yjang@kaist.ac.kr}
\affiliation{%
  \institution{Korea Advanced Institute of Science and Technology}
  \city{Daejeon}
  \country{Republic of Korea}
  }

\renewcommand{\shortauthors}{Sung and Baek, et al.}

\begin{abstract}
In this paper, we report an industry-academia collaborative study on the distribution method of fashion products using an artificial intelligence (AI) technique combined with an optimization method.
To meet the current fashion trend of short product lifetimes and an increasing variety of styles, the company produces limited volumes of a large variety of styles. However, due to the limited volume of each style, some styles may not be distributed to some off-line stores. As a result, this high-variety, low-volume strategy presents another challenge to distribution managers.
We collaborated with KOLON F/C, one of the largest fashion business units in South Korea, to develop models and an algorithm to optimally distribute the products to the stores based on the visual images of the products. The team developed a deep learning model that effectively represents the styles of clothes based on their visual image. Moreover, the team created an optimization model that effectively determines the product mix for each store based on the image representation of clothes.
In the past, computers were only considered to be useful for conducting logical calculations, and visual perception and cognition were considered to be difficult computational tasks.
The proposed approach is significant in that it uses both AI (perception and cognition) and mathematical optimization (logical calculation) to address a practical supply chain problem, which is why the study was called ``Breaking Moravec's Paradox.''
\end{abstract}




\maketitle

\section{Introduction}
KOLON Fashion and Culture (KOLON F/C), one of the top three fashion retailers in South Korea, oversees more than 10 fashion brands, and its signature brand KOLON Sport is one of the most popular outdoor brands in South Korea. Since it first began operating in the 1950s, a core value of KOLON F/C has been to understand its customers' needs by maintaining close relationships with its customer base. This core value has become the key philosophy of the firm's business operations. Then, KOLON F/C continues to operate small boutique type stores in many different locations. In addition to selling products, the store managers have been trained to build long-term relationships with their customers. This is one of the reasons the company has survived in business for more than 60 years.
Moreover, although the company has been expanding its sales in on-line channels, because the company values the direct interaction with customers, the off-line businesses continue to play the central role in the company's overall business operations.

In response to the ongoing advances in IT, which have driven constant and varying customer demand, fashion retailers have begun selling smaller volumes of a diverse range of products to accommodate the short selling seasons in which customer demand needs to be quickly and accurately satisfied \citep{bhardwaj2010fast}.
To respond to the changing business environment in the fashion retail industry, KOLON F/C has been aggressively adopting advanced information technologies and innovative operational concepts. In early 2015, KOLON F/C teamed up with the Korea Advanced Institute of Science and Technology (KAIST) and initiated the \textit{Smart Fashion Retail} project to create and develop advanced solutions for fashion retail operations.
The mission of the Smart Fashion Retail project is as follows:
\begin{enumerate}
\item	to improve the operation of KOLON F/C's off-line stores while maintaining the customer-centric philosophy, and
\item	to identify operations that produce unnecessary waste and solve the problems using advanced technologies.
\end{enumerate}
The team's ultimate goal is to upgrade the company's current ``Smart" operations, processes, and structures by creating a platform in which people can collaborate with machines and algorithms in the ``Second Machine Age" \citep{brynjolfsson2014second}. The team has designated the mission: ``Breaking \textit{Moravec's Paradox}." The mission stems from the well-known \textit{Moravec's Paradox} proposed by Hans Moravec and others in the 1980s, which states that \textit{it is comparatively easy to make computers exhibit adult level performance on intelligence tests or playing checkers, and difficult or impossible to give them the skills of a one-year-old when it comes to perception and mobility.}

The essence of the paradox is that while formal mathematical problems are difficult for humans to solve but relatively straightforward for computers, the problems that humans can solve intuitively, such as recognizing faces, are a constant challenge for computers.
However, the recently advanced AI algorithms, such as deep learning (DL), have made the computer processing of perceptual and cognitive information as reliable as human-level by learning with GPU and the large amount of data.

Visual information is particularly critical in the fashion business. The visual information about fashion items, such as clothes or footwear, is the key feature of the products.
However, it is very difficult to articulate the identification and quantification of these features in a formal manner. As a result, these features are not systematically used in most business operations. For instance, when products are distributed to KOLON F/C's off-line stores, it is important to send a wide variety of products to each store. In the past, product and distribution managers would select the product mix by manually inspecting each product. However, these naked-eye visual inspection and manual mix match processes are very subjective and time consuming. Moreover, selecting the product mix for specific stores while meeting the distribution requirements, such as \textit{``at least five different styles of T-shirts need to be distributed to the store in this area,"} involves not only the cognitive task of identifying the visual image, but also the conventional optimization of the logistical distribution.

The proposed study, performed in 2017, has applied the advanced AI-based approaches to the visual identification and representation of fashion products, and used the conventional optimization modeling techniques for the logical distribution. That is, the team has combined low level perception (image processing) and high level intelligence (optimization) technologies, which is why the paper is named ``Breaking \textit{Moravec's Paradox}." This was the last work of the project.

The remainder of this paper is organized as follows. Section \ref{sec:problem-scope} describes the project scope and the terms and definitions of the problem. Section \ref{sec:visual-similarity} explains the DL model used to quantify the visual similarities between styles. Section \ref{sec:optimization} presents the distribution optimization model based on the recognition of visual similarity. Section \ref{sec:result} provides an illustrative example of the optimal distribution. Section \ref{sec:conclusion} describes the implementation of the model and its business impact, and concludes by outlining the business and academic implications of the project.

\section{Problem, terms, and definitions}
\label{sec:problem-scope}

The Smart Retail team primarily focused on the distribution process for the off-line KOLON F/C stores because it was considered to be a key area in which significant improvements could be made.
KOLON F/C operates two central warehouses in South Korea, which are responsible for distributing products to hundreds of different off-line stores in the domestic market.
KOLON has a central distribution-inventory control policy whereby the distribution decisions in terms of which products are to be sent to which stores for the season are made by the distribution management team at the headquarters of each brand. The headquarters collect any inventory left over at the end of the season, which they sell to secondary markets or scrap. That is, the central headquarters is responsible for the distribution decisions and costs of any remaining inventory.
Therefore, the project team estimated that improving the distribution/sales processes would directly impact the overall operations of the brands.

\begin{figure}
\centering
\includegraphics[width=8.5cm]{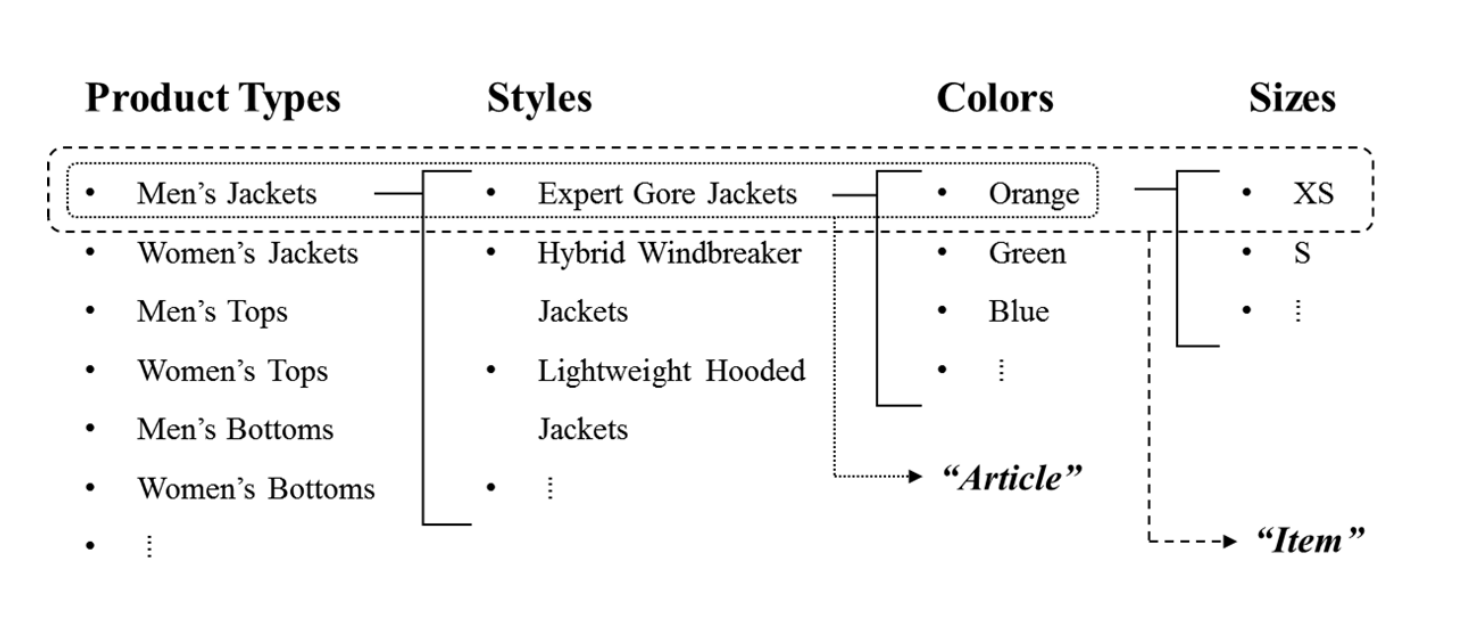}
\caption{The highest level of the product hierarchy is the \textit{product type}, such as men's t-shirts and women's bottoms. Each product type is broken down into \textit{styles}; for example, men's t-shirts include short-sleeve round-neck as a style. The styles are further broken down into \textit{colors}. A unit of a product line with a particular product type, style, and color is called an \textit{article}. Each article is further broken down into individual \textit{items}, distinguished by \textit{size}. Each product is then distinguished by the item level with its unique product type, style, color, and size.}
\label{fig:product-hierarchy}
\end{figure}

The recent fashion trends have also required companies to develop innovative distribution/sales processes.
As mentioned in the previous section, maintaining long-term customer relationships by interacting with the customer base in the off-line stores has been a core value of the KOLON F/C's operations.  However, the company is facing new challenges in its off-line operations.
Due to the current fast changing fashion trends, KOLON F/C needs to introduce many different styles of products at an unprecedented rate.
Note that KOLON F/C has a specific product hierarchy, as described in Figure \ref{fig:product-hierarchy}.
To meet this new fashion trend, KOLON F/C has also increased the number of new product styles.

Specifically, each season, hundreds of new products need to be introduced and distributed to each store. To meet the demands of this fast changing fashion trend, KOLON F/C has begun to introduce low volumes of many different styles of clothing. However, distributing high-variety, low-volume products to small boutique stores in many different locations has demanded a new system of distribution.
That is, each store only receives a limited range of the styles in each product type. Thus, the distribution managers need to provide an even mix of different styles. Note that the definition of the term ``different styles" is very abstract and is not easy to articulate in a logical form. Although it includes color, shape, and texture, it clearly refers to more than that.

\begin{figure}
    \centering
    \includegraphics[width=8.5cm]{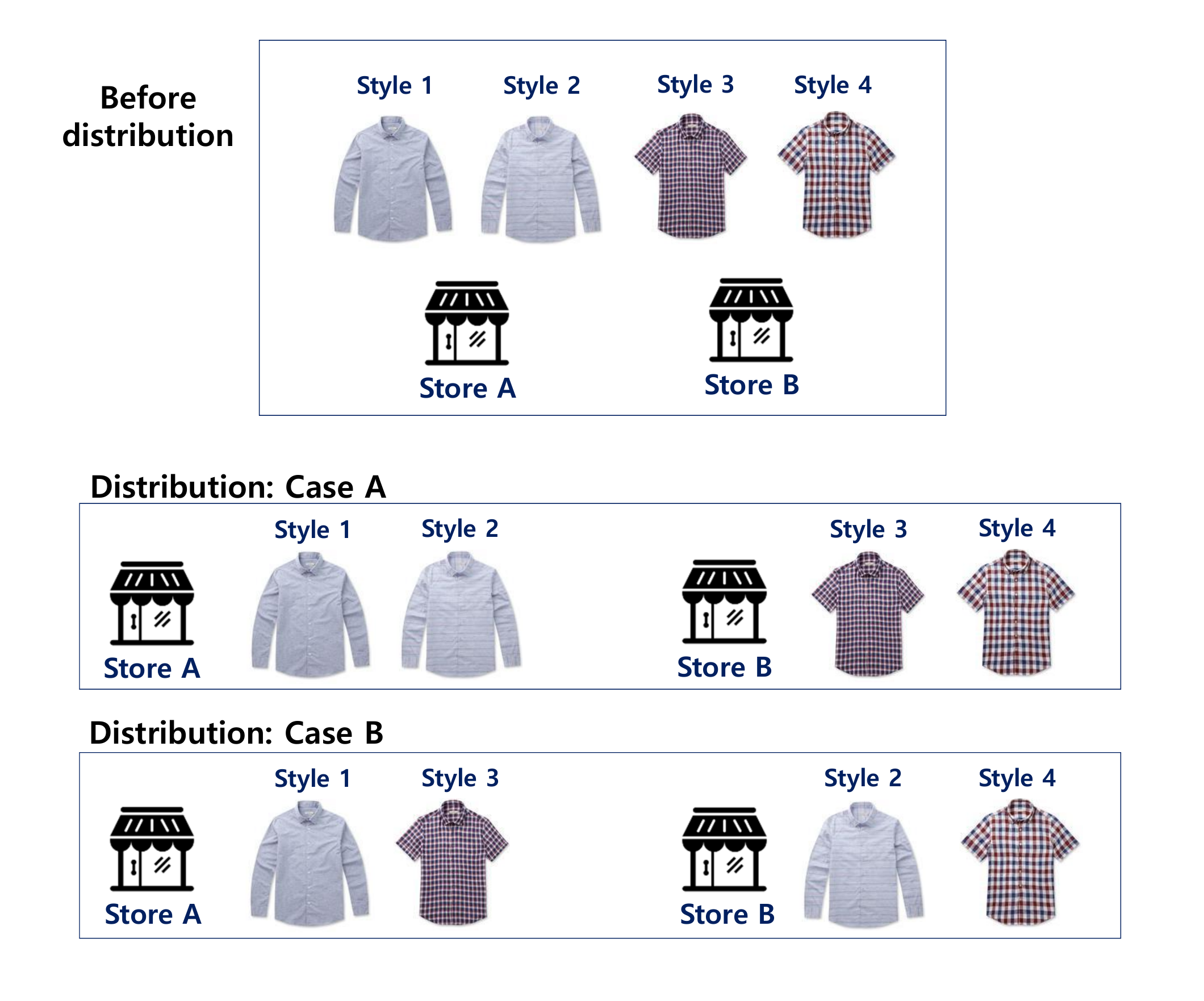}
    \caption{An illustrative example of the visual perception based distribution system in which four styles need to be distributed to two stores, and due to the limited quantity, each store can only receive two styles}
    \label{fig:distribution-problem-1}
\end{figure}

The distribution of an adequate style mix is critical. Figure \ref{fig:distribution-problem-1} presents an illustrative example of the distribution process. Suppose that four different styles need to be distributed to two stores (top panel). Due to the limited quantity of items available, only two different styles can be distributed to each store. In this example, the distribution in case B is preferable to that in case A. It is difficult to know in advance the product preferences of the customers who will visit the different stores. As a consequence, it is recommended to distribute as wide a variety of products to each store as possible. Therefore, case B is considered a better distribution because more heterogeneous mix of products is 
shown in case B than in case A.

The systematic distribution of a heterogeneous style mix presents two challenges. First, the characteristics of each style need to be identified. We call this process \textit{characteristic identification}. Second, the similarities and differences among the identified characteristics need to be quantified. We call this process the \textit{similarity measure}.

The distribution managers are required to distribute a heterogeneous style mix. Accordingly, the managers have tried to create style codes consisting of style categories and their attributes to identify the style characteristics of the products. For instance, the style of a T-shirt is defined according to the following categories: the length of the sleeves, shape of the neck collar, color pattern, main color, and number of buttons on the collar. Thus, a style code of (``short,'' ``round,'' ``solid,'' ``blue,'' ``2'') represents a T-shirt with short sleeves, a round neck, solid blue color, and two buttons.
However, the designers and product managers need to visually inspect the items, which has created numerous problems.
Due to the nature of the fashion products, it is not easy to map specific styles to specific codes. For example, suppose that a T-shirt has a check color pattern on the collar and sold color on the main body part. Should this T-shirt be classified as having a check or solid pattern? Moreover, there are thousands of possible colors and an infinite number of potential color patterns. Thus, it is impossible to classify each color or pattern in terms of finite attributes.
Moreover, KOLON F/C normally introduces more than 100 different styles in each product type each season, which makes it extremely cumbersome to manually define the attributes in each category with visual inspection.

Another problem is determining the similarities between different codes. Is a (``red,'' ``long sleeve,'' ``V-neck'') shirt more similar to a (``blue,'' ``long sleeve,'' ``round neck'') or an (``ocean blue,'' ``short sleeve,'' ``round neck'')? This question is very difficult to answer.

The manual code classification approach has not helped managers achieve the distribution goals and has generated many complaints from the store managers. For instance, one of the store managers the product team interviewed complained that ``My store received too many reddish jackets this season! We need more variety in colors.'' This issue has been known for years and managers have also known that the ad-hoc process has resulted in significant losses of sales and high inventory imbalance costs.
The project team decided to resolve this problem by integrating DL-based image representation technology with an optimization approach.
Specifically, the DL technology is used to identify the style characteristics and the optimization approach is used to determine the distribution of products considering the similarity measures.

Despite its practical importance, little progress has been made in addressing the distribution problem in the fashion and apparel industry. \cite{caro2010zara} developed a fashion distribution model for the global fashion brand
Zara, which incorporates the demand model and inventory replenishment model during the selling
season and considers the store policies in reducing the stocks for display. \cite{gallien2015initial} considered the initial shipment decision in the Zara distribution model while updating the demand forecast using a data-driven approach. Our project team also examined market forecasting, inventory optimization, and box packaging optimization, such as \citep{tiwari2007asp} and \citep{woong2017business}. However, these approaches do not consider the visual image or characteristics of the products in the distribution process. A number of studies have examined the use of visual image recognition and deep neural network based characteristic identification in a fashion context, such as \citep{bracher2016fashion}, \citep{hadi2015buy}, and \citep{zoghbi2016fashion}. However, these studies focus on the use of machine vision and network design for product recommendations rather than distribution. To the best of our knowledge, our paper is the first to apply DL image recognition to the distribution problem in the fashion industry. Accordingly, this paper focuses on the problem of off-line distribution using image embedding and optimization.

\section{Perception and cognition: quantifying visual similarity} \label{sec:visual-similarity}

\subsection{Kolon Smart Net}

\begin{figure}
    \centerline{\includegraphics[width=13cm]{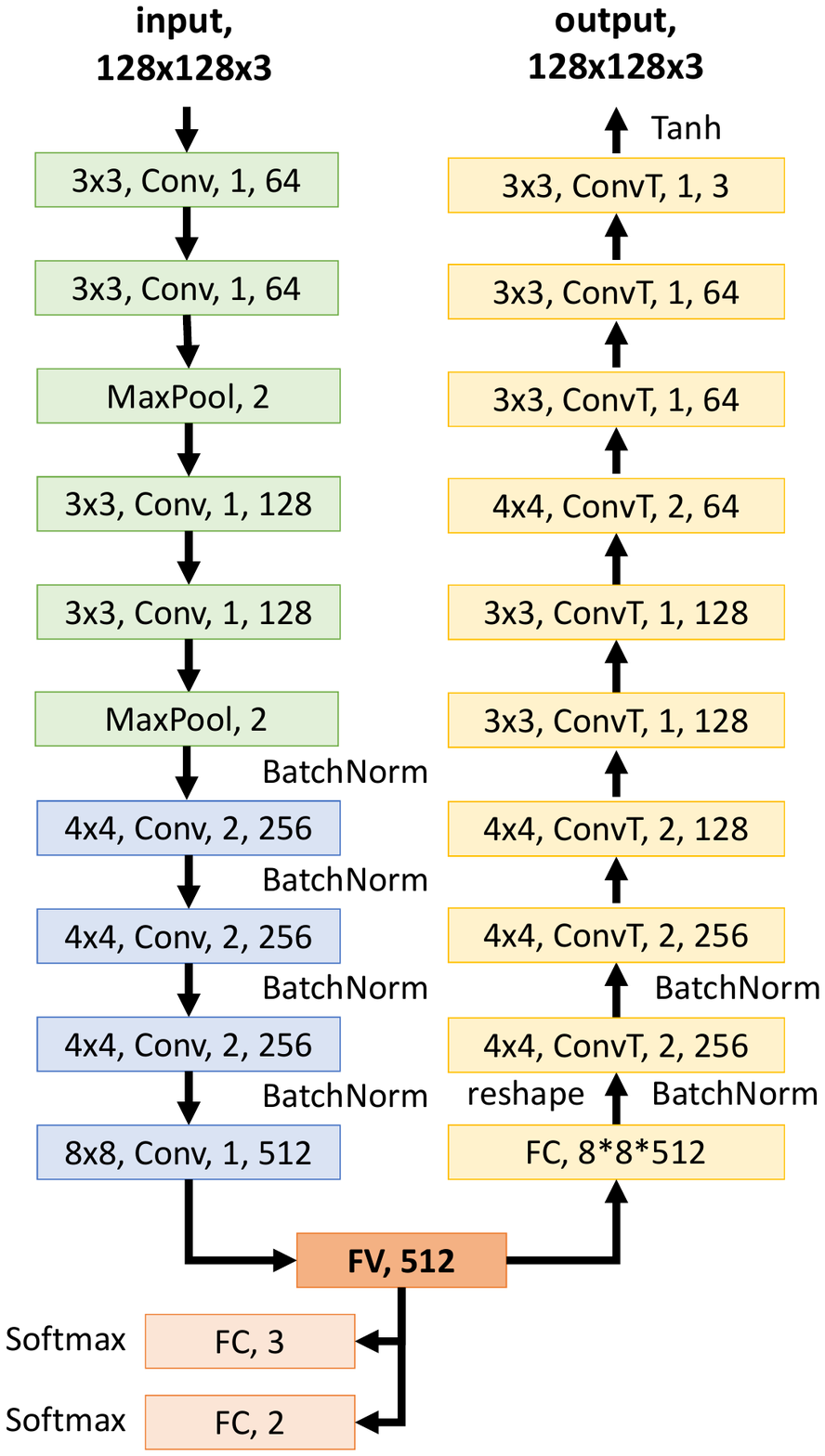}}
    \caption{The network structure of the Kolon Smart Net (KSN). The text in a box represents kernel size, type, stride, and channels. An encoded vector characterizing the product is called feature vector.}
    \label{fig:network-structure}
\end{figure}

We developed a deep convolutional autoencoder called the ``Kolon-Smart-Net'' (KSN) to effectively extract the features of fashion products from their visual images, as shown in Figure \ref{fig:network-structure}.
The KSN has the following characteristics:
\begin{itemize}
         \item Transfer learning
         \item Unsupervised learning
         \item Fashion feature layer
\end{itemize}
In the study, we used around 1300 original and clean images of fashion items (shirts and T-shirts), provided by KOLON F/C. To train the KSN with the small amount of images, we applied the \textit{transfer learning} approach, which is a machine learning method in which a pre-trained model is used as the starting point or a part of the learning structure for the target task. First four convolutional layers of a pre-trained VGGNet (by ImageNet dataset) is used for the first four convolutional layers of the encoding part of the KSN. Of course, we also performed image augmentation.

Because there are no specific and clear answers to the question of which set of attributes, such as color, pattern, collar, and line style, can represent the overall visual information of fashion items, it is difficult to apply a supervised learning strategy, which would require the class labels of the images to learn a network. Therefore, we developed the KSN based on a convolutional autoencoder (CAE), which is an unsupervised learning strategy \citep{masci2011stacked} for extracting the features of images without the use of classified and labeled data. In general, a CAE consists of two main parts: the encoder and the decoder. The encoder, which works on the convolutional layers, compresses an input image into a numerical feature, whereas the decoder, working on the transposed convolutional layers, reconstructs an image from the feature. The network is then trained to minimize the reconstruction loss, which is the error of the pixel values between the input image and the output image.

In addition to the CAE structure, we added classification layers that distinguish the major attributes of fashion products, to extract better feature vectors, as shown in the literature \citep{makhzani2015adversarial, jia2016learning}. Here, a classification layer called the fashion-feature layer is included to effectively extract the features of the images of the fashion products. As recommended by fashion designers and experts, fashion-feature layers extracts features that normally characterize fashion products. We chose the shape of collar and the length of sleeve as extracted features for shirts and T-shirts. With this information, we designed the KSN so that it could effectively identify these key features. The details of the structure of the KSN are provided in the Appendix \ref{app:KSN}.

We trained our network in Tensorflow, the famous DL library, and used a machine equipped with 32GB RAM, NVIDIA Titian X Pascal GPU, and a six core CPU \citep{abadi2016tensorflow}.

\begin{figure} \label{fig:KSN}
    \centerline{\includegraphics[width=11.5cm]{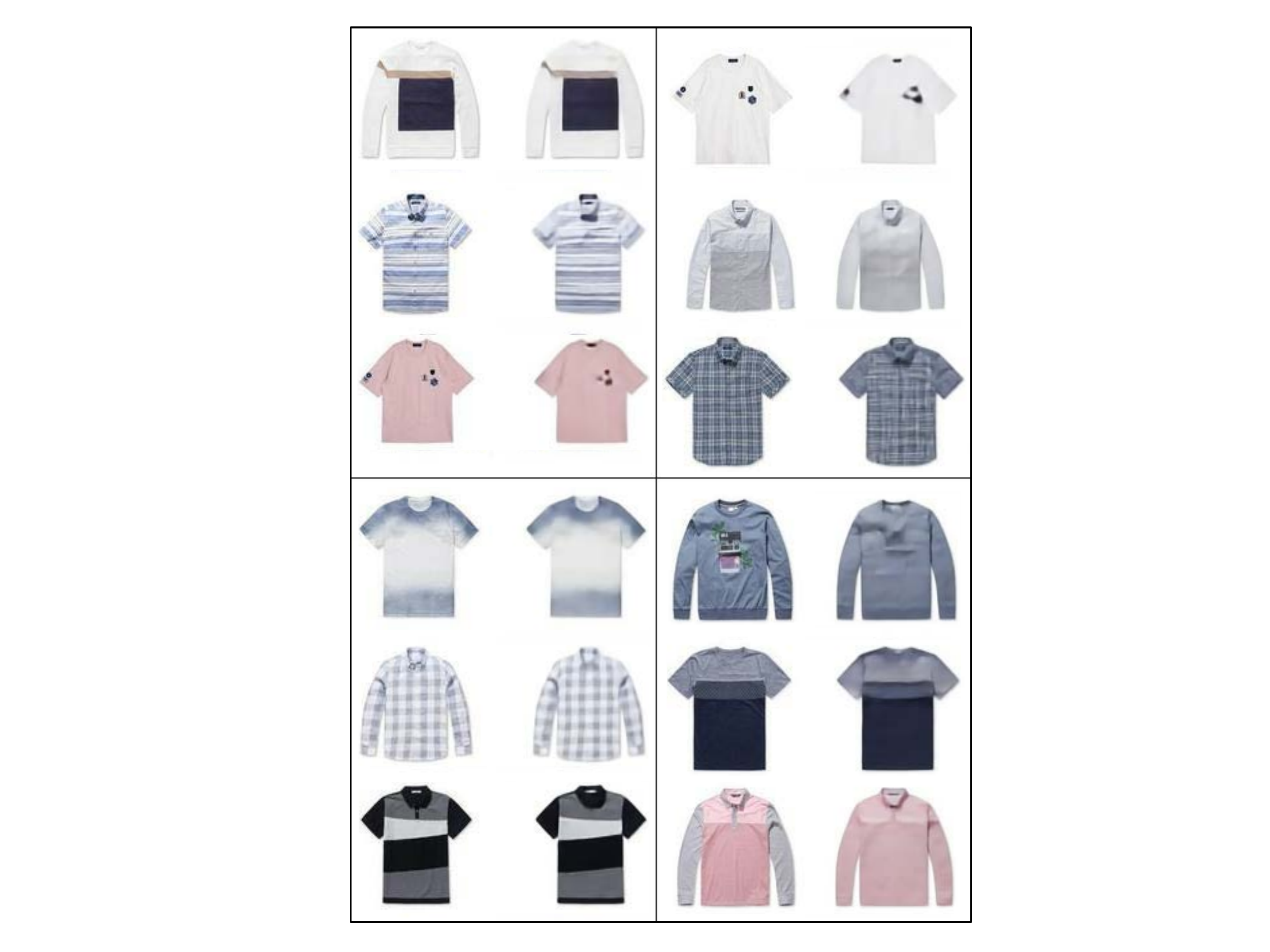}}
    \caption{Reconstructed image from the KSN}
    \label{fig:reconstructed-image}
\end{figure}

\subsection{Training Result}

\begin{figure}
\centering
\includegraphics[width=8.5cm]{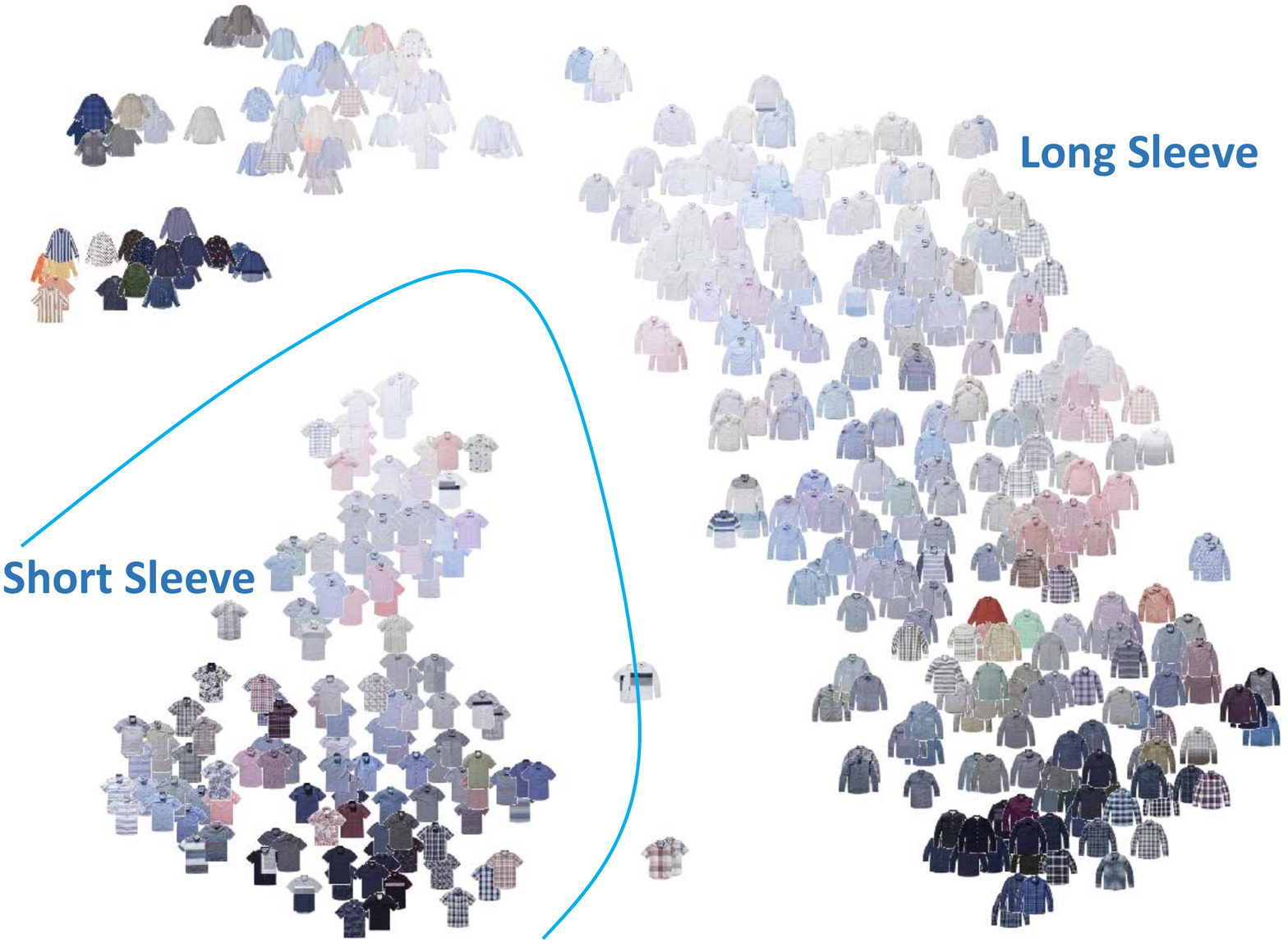}
\caption{The shirts are represented and laid out in a two-dimensional space}
\label{fig:visual-total}
\end{figure}

\begin{figure}
\centering
\includegraphics[width=8.5cm]{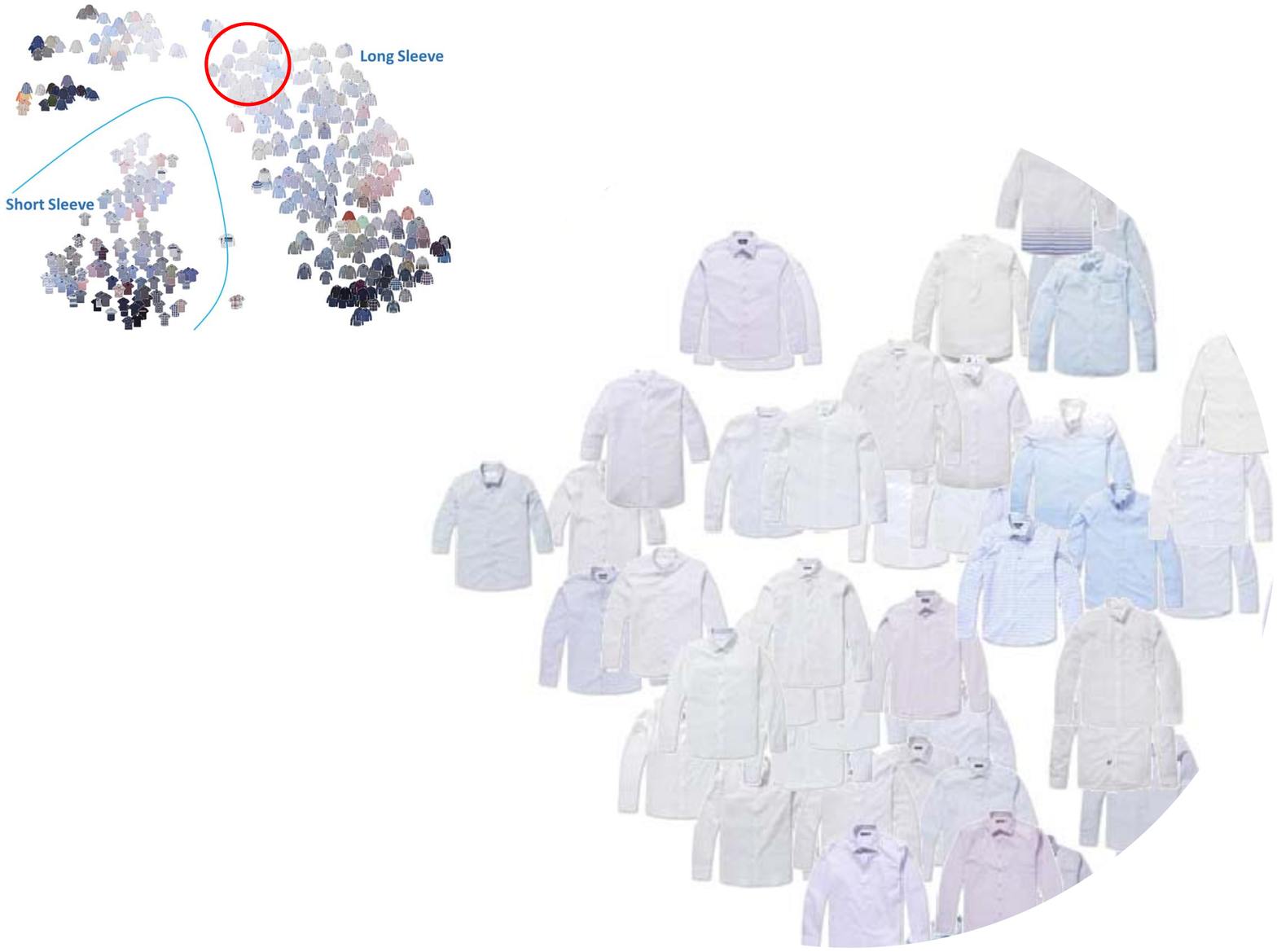}
\caption{An amplified view of Figure \ref{fig:visual-total}. The pastel colored shirts are located close to each other}
\label{fig:visual-pastel}
\end{figure}

\begin{figure}
\centering
\includegraphics[width=8.5cm]{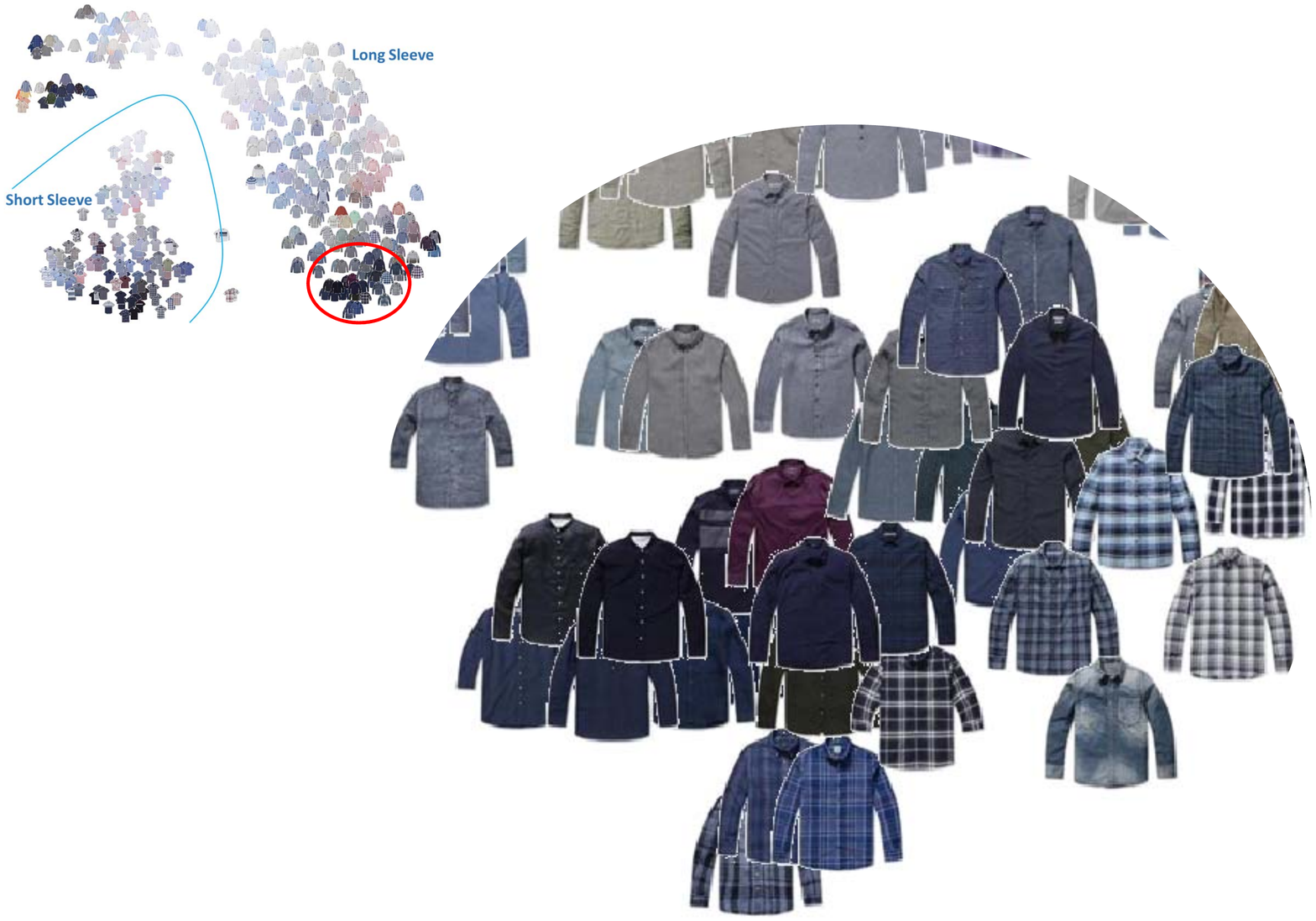}
\caption{An amplified view of Figure \ref{fig:visual-total}. The dark colored or dark patterned shirts are located close to each other}
\label{fig:visual-dark}
\end{figure}

The KSN takes a two-dimensional image of a product as input and generates a feature vector with 512 numerical elements as the output. Due to the nature of DL, it is impossible to interpret the meaning of the vector. However, if the network is well-trained for the extraction, the feature vectors will contain compact visual information for the reconstruction, and can be used to measure the similarity of the styles. So, firstly, we compared the reconstructed images to the original input images, as in Figure \ref{fig:reconstructed-image}. In Figure \ref{fig:reconstructed-image}, the first and third columns are the original images and the second and fourth columns are the reconstructed images. In detail, first column is for training images and third column is for validation images. Most of the images were reconstructed well except for some of the clothes with hard patterns. Secondly, for training and validation images of shirts, we plotted their feature vectors using a two-dimensional mapping method called t-SNE \citep{maaten2008visualizing}. As shown in Figure \ref{fig:visual-total}, the styles are mainly separated into two groups. Those on the left-bottom have short sleeves, whereas those on the other side have long sleeves. It can be seen that items with similar styles are located close to each other, whereas those with different styles are located at a distance (Figure \ref{fig:visual-pastel} and Figure \ref{fig:visual-dark}).

\section{Logical part: distribution optimization} \label{sec:optimization}

\subsection{Variety Measure}

\begin{figure}
\centering
\includegraphics[width=7.5cm]{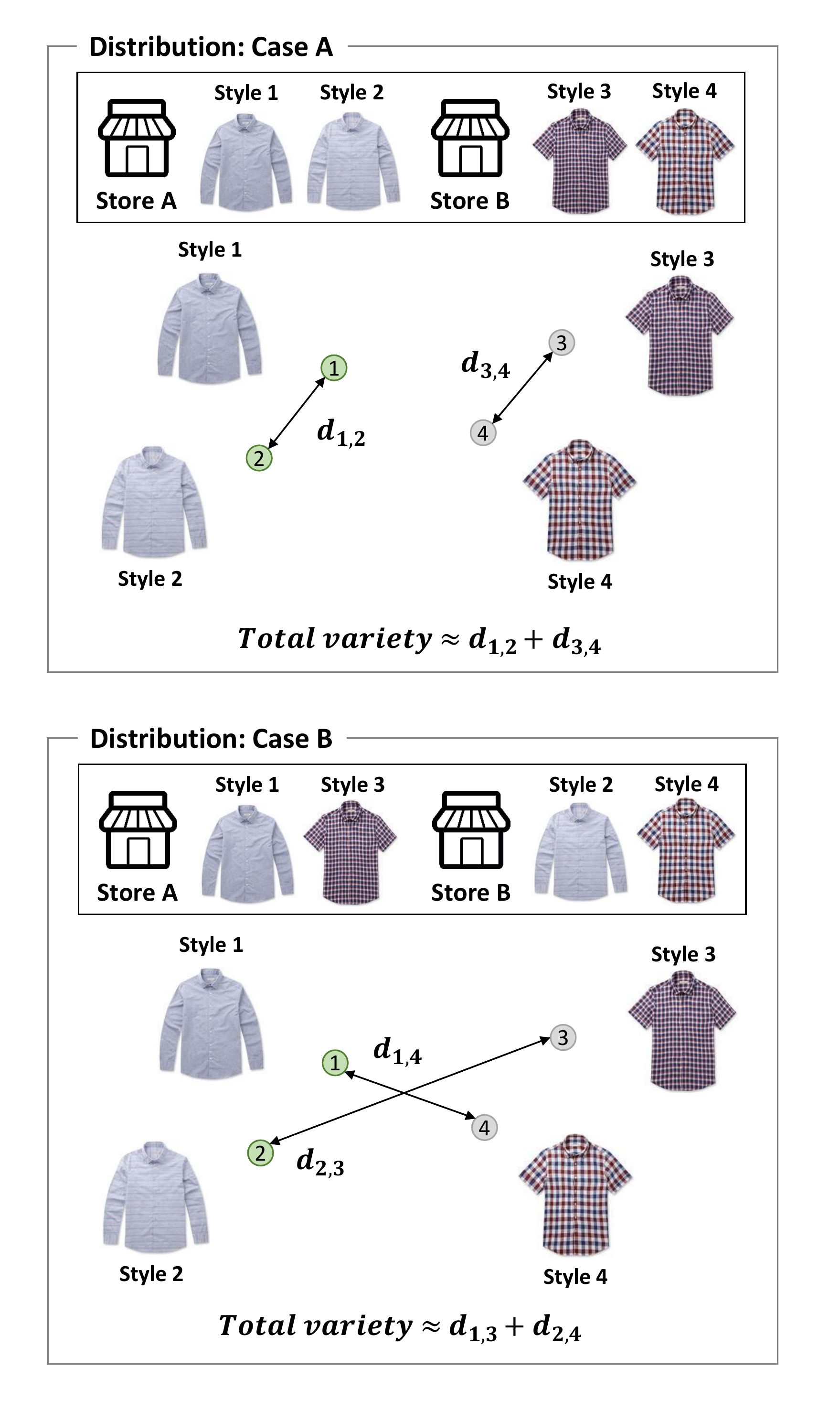}
\caption{Illustrative example of the variety measure: the variety measure for Case B is significantly greater than that for Case A}
\label{fig:distance}
\end{figure}

The feature vectors for each style, such as the shirts illustrated in Figure \ref{fig:visual-total}, are used to quantify the differences between the styles. Figure \ref{fig:distance} shows how the variety measures of the two different distribution cases in Figure \ref{fig:distribution-problem-1} are evaluated. The circles in each case illustrate the relative locations of the styles. For case A, style 1 and style 2 are distributed to store A, and styles 3 and 4 are distributed to store B. The visual variety of store A is then determined by the distance between styles 1 and 2 ($d_{1,2}$). Similarly, the visual variety of store B is the distance between styles 3 and 4 ($d_{3,4}$). The total visual variety of case A is then proportional to the sum of the two visual variety measures of the stores ($d_{1,2} + d_{3,4}$). For case B, the total visual variety is proportional to the distance between styles 1 and 4 and styles 2 and 3 ($d_{1,4} + d_{2,3}$). It can be clearly observed that the variety measure of case B is greater than that of case A.

As shown in Figure \ref{fig:distance}, the visual difference between two styles is evaluated from the distance between two vectors of the styles. Specifically, suppose $v(i)$ is the feature vector of style $i$. Then, the visual difference (distance) $d_{ij}$ between styles $i$ and $j$ is defined as $d_{ij} \overset{\Delta}{=}
||v(i)-v(j)||^{2}$.
The goal of the heterogeneous style mix is then translated to the process of selecting sets of styles that maximize the distances. This is known as the maximum dispersion problem, in which a given set of objects has to be partitioned into a number of groups such that all of the objects assigned to the same group are dispersed as much as possible with respect to some distance measure between the pairs of objects \citep{fernandez2013maximum} \citep{hassin1997approximation}.
In our case, the given set is the styles of the products, and the groups are the stores. We need to select specific styles for each store to maximize the dispersal of the styles in each store.
We now introduce the store variety measure ($v_s$) for quantifying the variety (the degree of heterogeneous mix of styles) for each store.
The five commonly used distance-based measures are \emph{MaxSumSum}, \emph{MaxMin}, \emph{MaxMinSum}, \emph{MaxSumMin}, and \emph{MaxMean}. (The prefix \emph{Max-} indicates that the measures should
be maximized to obtain the greatest variety of styles.) With these measures, the variety $v_s$ of a given product set $I_s$ for store $s$ is formulated as follows.

\begin{itemize}
\item \emph{MaxSumSum}: $v_{s} = \sum_{i,j \in I_{s}, i<j} d_{ij}$
\item \emph{MaxMin}:    $v_{s} = \min_{i,j \in I_{s}, i<j} d_{ij}$
\item \emph{MaxMinSum}: $v_{s} = \min_{i \in I_{s}} \sum_{j \in I_{s}\setminus{i}} d_{ij}$
\item \emph{MaxSumMin}: $v_{s} = \sum_{i \in I_{s}} \min_{j \in I_{s}\setminus{i}} d_{ij}$
\item \emph{MaxMean}:   $v_{s} = \sum_{i,j \in I_{s}, i<j} d_{ij}/|I_{s}|$
\end{itemize}

\emph{MaxSumSum} measures the variety of a given style set by summing the distances between
the possible pairs of styles in the store. \emph{MaxMin} measures the variety based on the minimum
distances between all of the possible pair of styles. \emph{MaxMinSum} measures the variety based on the minimum of the
sum of the distances between each individual style and all of the others. \emph{MaxSumMin}
measures the variety by summing the minimum distances between each individual style and all of the
others. \emph{MaxMean} measures the variety by dividing the sum of all of the distances
between the styles by the size of the style set.

We identified the following conditions for selecting an appropriate dispersion measure from these measures.

\begin{itemize}
\item Monotonicity: Suppose that a specific set of styles is selected for a store. Then the variety measure needs to be non-decreasing if an additional style is to be added to the existing set.
\item Linearity: When a specific number of styles are randomly selected, the variety measure should be linearly proportional to the number of selected styles.
\end{itemize}

The first condition is logically rigorous, while the second condition is somewhat qualitative. The first condition indicates that if any additional style is added to the existing set of styles in the store, the variety measure needs to be increased. This condition is illustrated by the example in Figure \ref{fig:variety-condition2}.

\begin{figure}
    \centerline{\includegraphics[width=6cm]{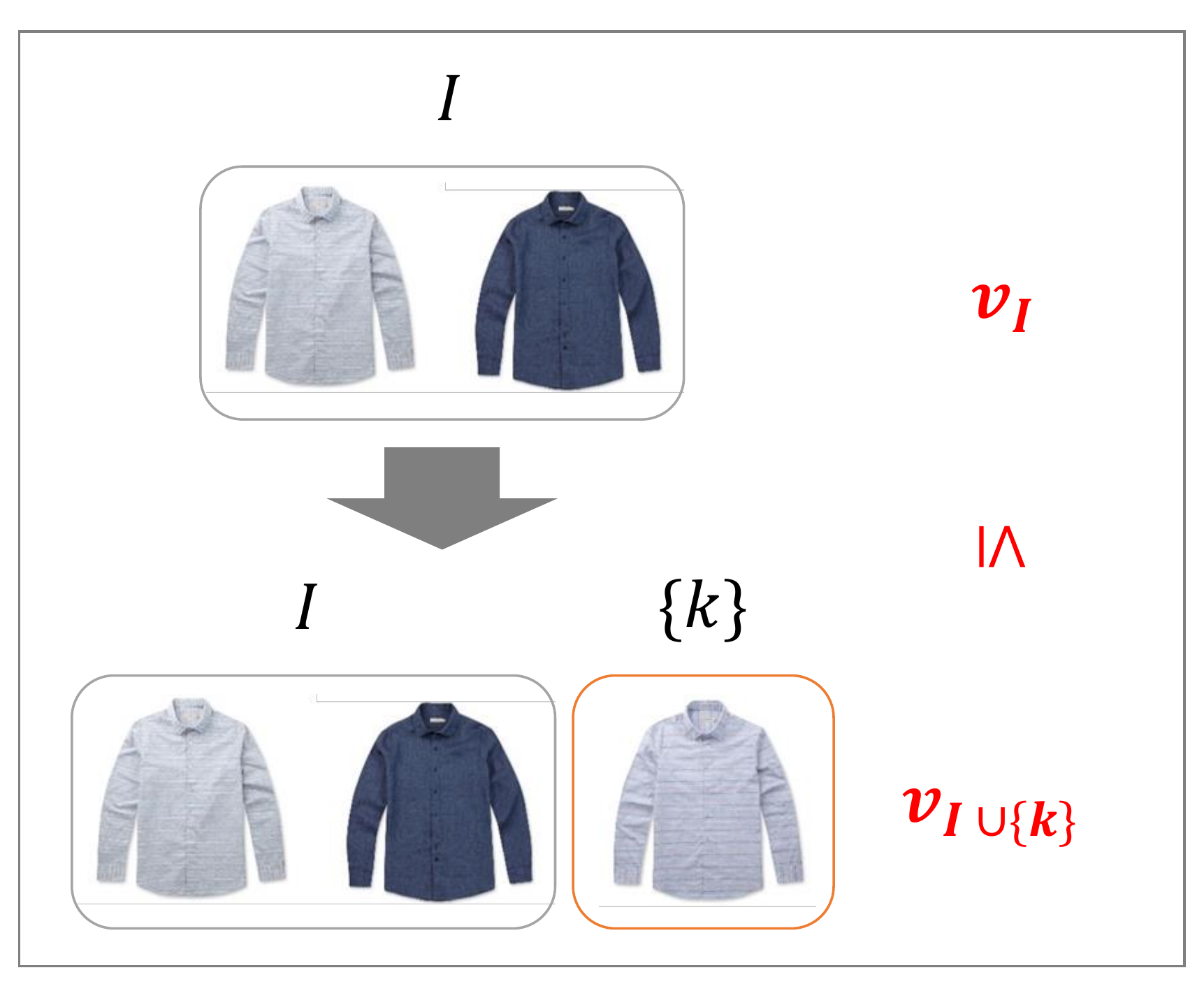}}
    \caption{Illustrative example of the first condition for the fashion variety measure.}
    \label{fig:variety-condition2}
\end{figure}

Appendix \ref{app:variety-measure} shows the observations and mathematical proofs indicating that \emph{MaxMean} is the variety measures satisfying the monotonicity condition, while \emph{MaxMinSum} and \emph{MaxSumMin} do not satisfy it. The numerical experiments in Appendix \ref{app:linearity} show that \emph{MaxMean} satisfies the linearity condition, while \emph{MaxMin} \emph{MaxSumSum} do not satisfy it. As a result, we use \emph{MaxMean} as the variety measure in our fashion variety model.

\subsection{Optimization Modeling}

We created an optimization model to optimally allocate the different styles of fashion products to the stores using the vector values of the product images generated from the KSN. Specifically, the model seeks to maximize the variety measure while satisfying the distribution requirements, such as the minimum distribution of styles for each store.

The distribution decision is subject to three types of managerial and operational constraints on which product and how much of it should be distributed to each store. First, the \emph{resource constraint} states that the total distribution quantity for product $i$ is restricted to its planned total quantity $p_i$. Second, the \emph{store quantity constraint} states that the deviation between the store desired quantity $q_s$ and the total distribution quantity for store $s$ should be less than or equal to $\alpha$\% of $q_s$. Finally, the \emph{minimum distribution quantity constraint} states that if a product is to be distributed to a store, then the quantity distributed to the store needs to be greater than $m_i$.
The optimization model is presented in Appendix \ref{app:optimization}.

\section{Experiment result} \label{sec:result}

\begin{figure}
    \centerline{\includegraphics[width=8.5cm]{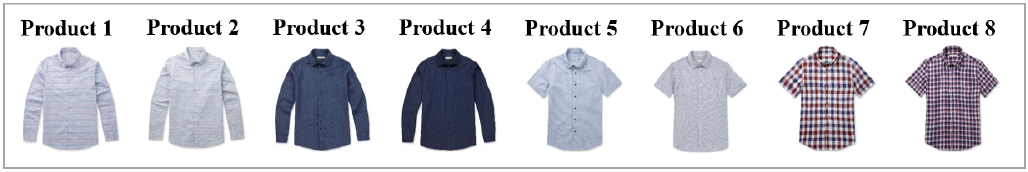}}
    \caption[Images of the test products]{Images of the test products.}
    \label{fig:example-products}
\end{figure}

The feature vectors are first extracted and the distance between each pair of styles is evaluated. The results are shown in Figure \ref{fig:example-distances} in the form of an intensity map table.

\begin{figure}
    \centerline{\includegraphics[width=8.5cm]{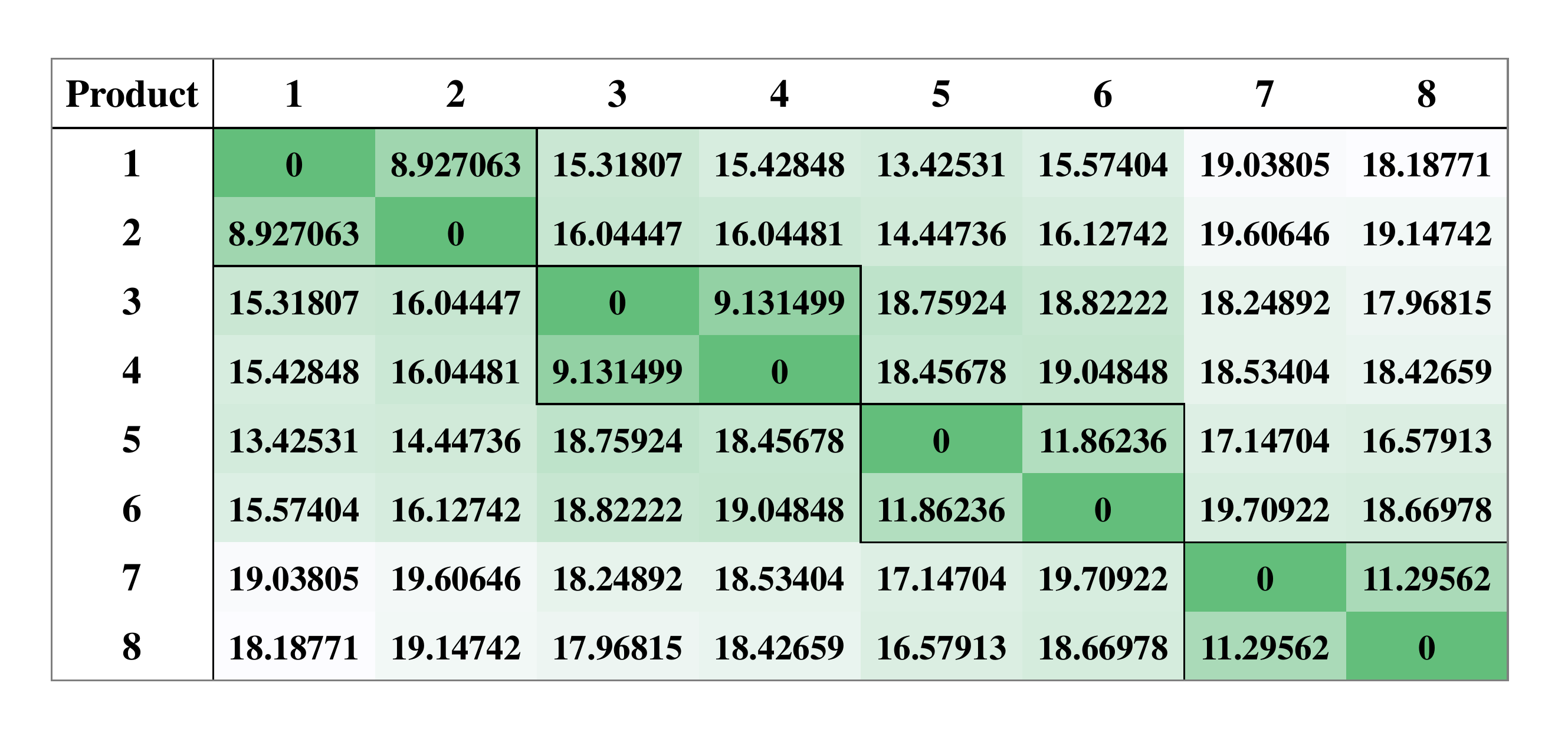}}
    \caption{Distance matrix of the test products.}
    \label{fig:example-distances}
\end{figure}

\begin{figure}
    \centerline{\includegraphics[width=8.5cm]{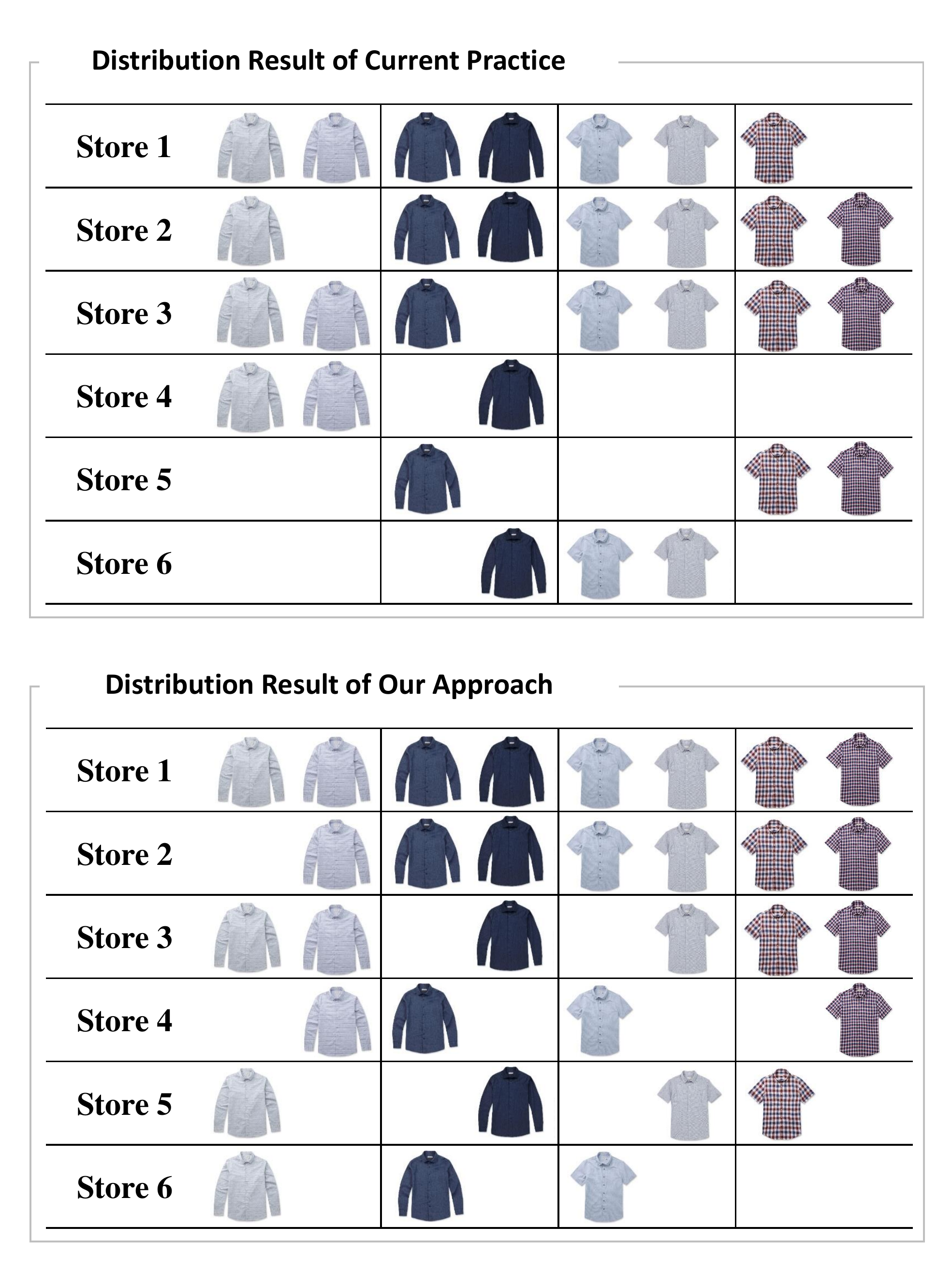}}
    \caption{Distribution results}
    \label{fig:example-result}
\end{figure}

This section presents an example case for illustrative purposes.
The target product type is men's shirts. As shown in Figure \ref{fig:example-products}, products 1-2, 3-4, 5-6, and 7-8 are visually similar.

The other parameters are setup as follows. The store desired quantity ranges between 10 and 30. For each product, we respectively set the planned available distribution quantity and the minimum distribution quantity to be 16 and 4, where the sum of all available distribution quantities of products is close to the sum of all desired quantities of stores. The percentage of maximum allowable deviation $\alpha$ is set to be 20\%.

Figure \ref{fig:example-result} shows the results of KOLON F/C's current practice and our approach, respectively. The product assignment for each store is listed by descending order of the store desired quantity. The stores with relatively higher expected sales (Stores 1 to 3) receive more products than those with relatively lower expected sales (Stores 4 to 6). The higher stores receive products from each distinct group in both results. It is particularly important to ensure the visual variety of products at low-sales store because they receive only a few different products. Note that each store except Store 6 receives at least one product in each group. Moreover, each product that Store 6 receives is from a distinct group. In terms of the total variety, our approach (218.77) is about 13\% better than the current practice (193.51).

\section{Conclusion} \label{sec:conclusion}
The high-variety and low-volume strategy of production in fashion industry ensues the dispersion problem of the styles of fashion products. To solve this practical supply chain problem, we proposed the approach that it uses both AI based DL technology for image representation and mathematical optimization with the variety measure. The limitation of this study is showing only a small case in the experiment. So, by on-going research, we have tried to develop a heuristic optimization algorithm and test it for a real industrial case with more than 30 styles and 100 stores. 

The academic contribution of our work is clear. The KSN is one of the first cases in which cutting-edge AI based DL technology has been combined with the conventional optimization modeling approach. Although the DL approach is rapidly improving and it has been widely applied in the engineering field, there have been few applications of the technology in the business arena. This paper shows that the current AI technology can be effectively combined with the conventional optimization modeling approach to create new horizons in business analytics.

\begin{acks}
Thank KOLON F/C for supporting the research.
\end{acks}

\bibliographystyle{ACM-Reference-Format}
\bibliography{sample-base}


\begin{thebibliography}{19}


\ifx \showCODEN    \undefined \def \showCODEN     #1{\unskip}     \fi
\ifx \showDOI      \undefined \def \showDOI       #1{#1}\fi
\ifx \showISBNx    \undefined \def \showISBNx     #1{\unskip}     \fi
\ifx \showISBNxiii \undefined \def \showISBNxiii  #1{\unskip}     \fi
\ifx \showISSN     \undefined \def \showISSN      #1{\unskip}     \fi
\ifx \showLCCN     \undefined \def \showLCCN      #1{\unskip}     \fi
\ifx \shownote     \undefined \def \shownote      #1{#1}          \fi
\ifx \showarticletitle \undefined \def \showarticletitle #1{#1}   \fi
\ifx \showURL      \undefined \def \showURL       {\relax}        \fi
\providecommand\bibfield[2]{#2}
\providecommand\bibinfo[2]{#2}
\providecommand\natexlab[1]{#1}
\providecommand\showeprint[2][]{arXiv:#2}

\bibitem[\protect\citeauthoryear{Abadi, Agarwal, Barham, Brevdo, Chen, Citro,
  Corrado, Davis, Dean, Devin, et~al\mbox{.}}{Abadi et~al\mbox{.}}{2016}]%
        {abadi2016tensorflow}
\bibfield{author}{\bibinfo{person}{Mart{\'\i}n Abadi}, \bibinfo{person}{Ashish
  Agarwal}, \bibinfo{person}{Paul Barham}, \bibinfo{person}{Eugene Brevdo},
  \bibinfo{person}{Zhifeng Chen}, \bibinfo{person}{Craig Citro},
  \bibinfo{person}{Greg~S Corrado}, \bibinfo{person}{Andy Davis},
  \bibinfo{person}{Jeffrey Dean}, \bibinfo{person}{Matthieu Devin},
  {et~al\mbox{.}}} \bibinfo{year}{2016}\natexlab{}.
\newblock \showarticletitle{Tensorflow: Large-scale machine learning on
  heterogeneous distributed systems}.
\newblock \bibinfo{journal}{\emph{arXiv preprint arXiv:1603.04467}}
  (\bibinfo{year}{2016}).
\newblock


\bibitem[\protect\citeauthoryear{Bhardwaj and Fairhurst}{Bhardwaj and
  Fairhurst}{2010}]%
        {bhardwaj2010fast}
\bibfield{author}{\bibinfo{person}{Vertica Bhardwaj} {and} \bibinfo{person}{Ann
  Fairhurst}.} \bibinfo{year}{2010}\natexlab{}.
\newblock \showarticletitle{Fast fashion: response to changes in the fashion
  industry}.
\newblock \bibinfo{journal}{\emph{The International Review of Retail,
  Distribution and Consumer Research}} \bibinfo{volume}{20},
  \bibinfo{number}{1} (\bibinfo{year}{2010}), \bibinfo{pages}{165--173}.
\newblock


\bibitem[\protect\citeauthoryear{Bracher, Heinz, and Vollgraf}{Bracher
  et~al\mbox{.}}{2016}]%
        {bracher2016fashion}
\bibfield{author}{\bibinfo{person}{Christian Bracher},
  \bibinfo{person}{Sebastian Heinz}, {and} \bibinfo{person}{Roland Vollgraf}.}
  \bibinfo{year}{2016}\natexlab{}.
\newblock \showarticletitle{Fashion DNA: Merging content and sales data for
  recommendation and article mapping}.
\newblock \bibinfo{journal}{\emph{arXiv preprint arXiv:1609.02489}}
  (\bibinfo{year}{2016}).
\newblock


\bibitem[\protect\citeauthoryear{Brynjolfsson and McAfee}{Brynjolfsson and
  McAfee}{2014}]%
        {brynjolfsson2014second}
\bibfield{author}{\bibinfo{person}{Erik Brynjolfsson} {and}
  \bibinfo{person}{Andrew McAfee}.} \bibinfo{year}{2014}\natexlab{}.
\newblock \bibinfo{booktitle}{\emph{The second machine age: Work, progress, and
  prosperity in a time of brilliant technologies}}.
\newblock \bibinfo{publisher}{WW Norton \& Company}.
\newblock


\bibitem[\protect\citeauthoryear{Caro, Gallien, D{\'\i}az, Garc{\'\i}a,
  Corredoira, Montes, Ramos, and Correa}{Caro et~al\mbox{.}}{2010}]%
        {caro2010zara}
\bibfield{author}{\bibinfo{person}{Felipe Caro},
  \bibinfo{person}{J{\'e}r{\'e}mie Gallien}, \bibinfo{person}{Miguel
  D{\'\i}az}, \bibinfo{person}{Javier Garc{\'\i}a},
  \bibinfo{person}{Jos{\'e}~Manuel Corredoira}, \bibinfo{person}{Marcos
  Montes}, \bibinfo{person}{Jos{\'e}~Antonio Ramos}, {and}
  \bibinfo{person}{Juan Correa}.} \bibinfo{year}{2010}\natexlab{}.
\newblock \showarticletitle{Zara uses operations research to reengineer its
  global distribution process}.
\newblock \bibinfo{journal}{\emph{Interfaces}} \bibinfo{volume}{40},
  \bibinfo{number}{1} (\bibinfo{year}{2010}), \bibinfo{pages}{71--84}.
\newblock


\bibitem[\protect\citeauthoryear{Fern{\'a}ndez, Kalcsics, and
  Nickel}{Fern{\'a}ndez et~al\mbox{.}}{2013}]%
        {fernandez2013maximum}
\bibfield{author}{\bibinfo{person}{Elena Fern{\'a}ndez},
  \bibinfo{person}{J{\"o}rg Kalcsics}, {and} \bibinfo{person}{Stefan Nickel}.}
  \bibinfo{year}{2013}\natexlab{}.
\newblock \showarticletitle{The maximum dispersion problem}.
\newblock \bibinfo{journal}{\emph{Omega}} \bibinfo{volume}{41},
  \bibinfo{number}{4} (\bibinfo{year}{2013}), \bibinfo{pages}{721--730}.
\newblock


\bibitem[\protect\citeauthoryear{Gallien, Mersereau, Garro, Mora, and
  Vidal}{Gallien et~al\mbox{.}}{2015}]%
        {gallien2015initial}
\bibfield{author}{\bibinfo{person}{J{\'e}r{\'e}mie Gallien},
  \bibinfo{person}{Adam~J Mersereau}, \bibinfo{person}{Andres Garro},
  \bibinfo{person}{Alberte~Dapena Mora}, {and}
  \bibinfo{person}{Mart{\'\i}n~N{\'o}voa Vidal}.}
  \bibinfo{year}{2015}\natexlab{}.
\newblock \showarticletitle{Initial shipment decisions for new products at
  Zara}.
\newblock \bibinfo{journal}{\emph{Operations Research}} \bibinfo{volume}{63},
  \bibinfo{number}{2} (\bibinfo{year}{2015}), \bibinfo{pages}{269--286}.
\newblock


\bibitem[\protect\citeauthoryear{Hadi~Kiapour, Han, Lazebnik, Berg, and
  Berg}{Hadi~Kiapour et~al\mbox{.}}{2015}]%
        {hadi2015buy}
\bibfield{author}{\bibinfo{person}{M Hadi~Kiapour}, \bibinfo{person}{Xufeng
  Han}, \bibinfo{person}{Svetlana Lazebnik}, \bibinfo{person}{Alexander~C
  Berg}, {and} \bibinfo{person}{Tamara~L Berg}.}
  \bibinfo{year}{2015}\natexlab{}.
\newblock \showarticletitle{Where to buy it: Matching street clothing photos in
  online shops}. In \bibinfo{booktitle}{\emph{Proceedings of the IEEE
  international conference on computer vision}}. \bibinfo{pages}{3343--3351}.
\newblock


\bibitem[\protect\citeauthoryear{Hassin, Rubinstein, and Tamir}{Hassin
  et~al\mbox{.}}{1997}]%
        {hassin1997approximation}
\bibfield{author}{\bibinfo{person}{Refael Hassin}, \bibinfo{person}{Shlomi
  Rubinstein}, {and} \bibinfo{person}{Arie Tamir}.}
  \bibinfo{year}{1997}\natexlab{}.
\newblock \showarticletitle{Approximation algorithms for maximum dispersion}.
\newblock \bibinfo{journal}{\emph{Operations research letters}}
  \bibinfo{volume}{21}, \bibinfo{number}{3} (\bibinfo{year}{1997}),
  \bibinfo{pages}{133--137}.
\newblock


\bibitem[\protect\citeauthoryear{Jia, Huang, Shen, He, Liu, Luan, and Yan}{Jia
  et~al\mbox{.}}{2016}]%
        {jia2016learning}
\bibfield{author}{\bibinfo{person}{Jia Jia}, \bibinfo{person}{Jie Huang},
  \bibinfo{person}{Guangyao Shen}, \bibinfo{person}{Tao He},
  \bibinfo{person}{Zhiyuan Liu}, \bibinfo{person}{Huan-Bo Luan}, {and}
  \bibinfo{person}{Chao Yan}.} \bibinfo{year}{2016}\natexlab{}.
\newblock \showarticletitle{Learning to Appreciate the Aesthetic Effects of
  Clothing.}. In \bibinfo{booktitle}{\emph{AAAI}}. \bibinfo{pages}{1216--1222}.
\newblock


\bibitem[\protect\citeauthoryear{Maaten and Hinton}{Maaten and Hinton}{2008}]%
        {maaten2008visualizing}
\bibfield{author}{\bibinfo{person}{Laurens van~der Maaten} {and}
  \bibinfo{person}{Geoffrey Hinton}.} \bibinfo{year}{2008}\natexlab{}.
\newblock \showarticletitle{Visualizing data using t-SNE}.
\newblock \bibinfo{journal}{\emph{Journal of machine learning research}}
  \bibinfo{volume}{9}, \bibinfo{number}{Nov} (\bibinfo{year}{2008}),
  \bibinfo{pages}{2579--2605}.
\newblock


\bibitem[\protect\citeauthoryear{Makhzani, Shlens, Jaitly, Goodfellow, and
  Frey}{Makhzani et~al\mbox{.}}{2015}]%
        {makhzani2015adversarial}
\bibfield{author}{\bibinfo{person}{Alireza Makhzani}, \bibinfo{person}{Jonathon
  Shlens}, \bibinfo{person}{Navdeep Jaitly}, \bibinfo{person}{Ian Goodfellow},
  {and} \bibinfo{person}{Brendan Frey}.} \bibinfo{year}{2015}\natexlab{}.
\newblock \showarticletitle{Adversarial autoencoders}.
\newblock \bibinfo{journal}{\emph{arXiv preprint arXiv:1511.05644}}
  (\bibinfo{year}{2015}).
\newblock


\bibitem[\protect\citeauthoryear{Masci, Meier, Cire{\c{s}}an, and
  Schmidhuber}{Masci et~al\mbox{.}}{2011}]%
        {masci2011stacked}
\bibfield{author}{\bibinfo{person}{Jonathan Masci}, \bibinfo{person}{Ueli
  Meier}, \bibinfo{person}{Dan Cire{\c{s}}an}, {and}
  \bibinfo{person}{J{\"u}rgen Schmidhuber}.} \bibinfo{year}{2011}\natexlab{}.
\newblock \showarticletitle{Stacked convolutional auto-encoders for
  hierarchical feature extraction}.
\newblock \bibinfo{journal}{\emph{Artificial Neural Networks and Machine
  Learning--ICANN 2011}} (\bibinfo{year}{2011}), \bibinfo{pages}{52--59}.
\newblock


\bibitem[\protect\citeauthoryear{Odena, Dumoulin, and Olah}{Odena
  et~al\mbox{.}}{2016}]%
        {odena2016deconvolution}
\bibfield{author}{\bibinfo{person}{Augustus Odena}, \bibinfo{person}{Vincent
  Dumoulin}, {and} \bibinfo{person}{Chris Olah}.}
  \bibinfo{year}{2016}\natexlab{}.
\newblock \showarticletitle{Deconvolution and checkerboard artifacts}.
\newblock \bibinfo{journal}{\emph{Distill}} \bibinfo{volume}{1},
  \bibinfo{number}{10} (\bibinfo{year}{2016}), \bibinfo{pages}{e3}.
\newblock


\bibitem[\protect\citeauthoryear{Oquab, Bottou, Laptev, and Sivic}{Oquab
  et~al\mbox{.}}{2014}]%
        {oquab2014learning}
\bibfield{author}{\bibinfo{person}{Maxime Oquab}, \bibinfo{person}{Leon
  Bottou}, \bibinfo{person}{Ivan Laptev}, {and} \bibinfo{person}{Josef Sivic}.}
  \bibinfo{year}{2014}\natexlab{}.
\newblock \showarticletitle{Learning and transferring mid-level image
  representations using convolutional neural networks}. In
  \bibinfo{booktitle}{\emph{Proceedings of the IEEE conference on computer
  vision and pattern recognition}}. \bibinfo{pages}{1717--1724}.
\newblock


\bibitem[\protect\citeauthoryear{Simonyan and Zisserman}{Simonyan and
  Zisserman}{2014}]%
        {simonyan2014very}
\bibfield{author}{\bibinfo{person}{Karen Simonyan} {and}
  \bibinfo{person}{Andrew Zisserman}.} \bibinfo{year}{2014}\natexlab{}.
\newblock \showarticletitle{Very deep convolutional networks for large-scale
  image recognition}.
\newblock \bibinfo{journal}{\emph{arXiv preprint arXiv:1409.1556}}
  (\bibinfo{year}{2014}).
\newblock


\bibitem[\protect\citeauthoryear{Sung, Jang, Kim, and Lee}{Sung
  et~al\mbox{.}}{2017}]%
        {woong2017business}
\bibfield{author}{\bibinfo{person}{Shin~Woong Sung}, \bibinfo{person}{Young~Jae
  Jang}, \bibinfo{person}{Jung~Hoon Kim}, {and} \bibinfo{person}{Juyeong Lee}.}
  \bibinfo{year}{2017}\natexlab{}.
\newblock \showarticletitle{Business Analytics for Streamlined Assort Packing
  and Distribution of Fashion Goods at {Kolon Sport}}.
\newblock \bibinfo{journal}{\emph{Interfaces}} \bibinfo{volume}{47},
  \bibinfo{number}{6} (\bibinfo{year}{2017}), \bibinfo{pages}{555--573}.
\newblock


\bibitem[\protect\citeauthoryear{Tiwari and Gavirneni}{Tiwari and
  Gavirneni}{2007}]%
        {tiwari2007asp}
\bibfield{author}{\bibinfo{person}{Vikram Tiwari} {and}
  \bibinfo{person}{Srinagesh Gavirneni}.} \bibinfo{year}{2007}\natexlab{}.
\newblock \showarticletitle{ASP, the art and science of practice: Recoupling
  inventory control research and practice: guidelines for achieving synergy}.
\newblock \bibinfo{journal}{\emph{Interfaces}} \bibinfo{volume}{37},
  \bibinfo{number}{2} (\bibinfo{year}{2007}), \bibinfo{pages}{176--186}.
\newblock


\bibitem[\protect\citeauthoryear{Zoghbi, Heyman, Gomez, and Moens}{Zoghbi
  et~al\mbox{.}}{2016}]%
        {zoghbi2016fashion}
\bibfield{author}{\bibinfo{person}{Susana Zoghbi}, \bibinfo{person}{Geert
  Heyman}, \bibinfo{person}{Juan~Carlos Gomez}, {and}
  \bibinfo{person}{Marie-Francine Moens}.} \bibinfo{year}{2016}\natexlab{}.
\newblock \showarticletitle{Fashion meets computer vision and nlp at e-commerce
  search}.
\newblock \bibinfo{journal}{\emph{International Journal of Computer and
  Electrical Engineering}} \bibinfo{volume}{8}, \bibinfo{number}{1}
  (\bibinfo{year}{2016}), \bibinfo{pages}{31}.
\newblock


\end{thebibliography}

\appendix

\section{KSN structure} \label{app:KSN}

KSN is a deep convolutional autoencoder. The encoder part of the KSN is composed of 8 convolutional layers. To achieve stable training of the network, we made first four convolutional layers of the KSN the same as the first four convolutional layers of the VGG-16 network, and sat the same weights as the pre-trained VGG-16 at the initialization, as a form of transfer learning \citep{oquab2014learning, simonyan2014very}. During the initial epochs of training, we did not update the weights of the four layers with gradient, 
but after that, we updated them. We also used batch normalization and LeakyReLU activation in the later layers of the encoder. The encoder of the KSN then compresses a 128 $\times$ 128 $\times$ 3 image into a 512-dimension feature vector.

Next, the fashion-feature layers (classification output layers) are processed with softmax in the middle of the KSN. For the target styles (shirts and T-shirts), there are two fashion-feature layers for the length of the sleeve and presence of collar, respectively, which we selected as the main attributes of these garments. Finally, the decoder part of the KSN comprises one fully connected layer and nine transposed convolutional layers for upsampling the 512-dimension feature vector to the original 128 $\times$ 128 $\times$ 3 image. We also used batch normalization, LeakyReLU activation, and scheduled learning rates during the training. Additionally, we used 4 $\times$ 4 kernel to avoid the checkerboard artifacts when using stride 2 in a deconvolutional layer \citep{odena2016deconvolution}. We did the same behavior for a convolutional layer.

\section{Variety measure conditions}
\subsection{Monotonicity condition}\label{app:variety-measure}
We provide the following observations and theorems for investigating the
non-decreasing patterns of the variety measures (\emph{MaxMinSum}, \emph{MaxSumMin}, and \emph{MaxMean}), when any one product is added to the existing
product set.

\begin{obs}
\label{obs-MaxMinSum}
The variety of a product set by \emph{MaxMinSum} is not necessarily non-decreasing when any one
product is added to the existing product set.
\end{obs}

\begin{proof}
\label{pf_obs-MaxMinSum}
To prove the theorem, we give a counterexample on $v_{I} \leq v_{I\cup\{k\}}$ for \emph{MaxMinSum},
which is $\min_{i \in I} \sum_{j \in I\setminus{i}} d_{ij}$ $\leq \min_{i \in I\cup\{k\}}$ $\sum_{j \in I\cup\{k\}\setminus{i}} d_{ij}$. We consider a product set $I$ with three products, namely P1, P2, and P3, whose feature vectors form a regular triangle with a length of 1, as shown in Figure
\ref{fig:thm-MaxMinSum}. The variety of the product set $v_{I}$ by \emph{MaxMinSum} is $\min (2, 2,
2) = 2$. Suppose another product, P4, is added to the product set $I$, where P4 is located in the
incenter of the triangle. In this case, the variety of the product set $v_{I\cup\{P4\}}$ by
\emph{MaxMinSum} is $\min (2+1/\sqrt{3},2+1/\sqrt{3},2+1/\sqrt{3},3/\sqrt{3}) \approx 1.7321$. In
this example, \emph{MaxMinSum} decreases when a product P4 is added to the existing product set
$I$, which contradicts the third condition $v_{I} \leq v_{I\cup\{k\}}$.
\end{proof}

\begin{figure}
    \centerline{\includegraphics[width=6cm]{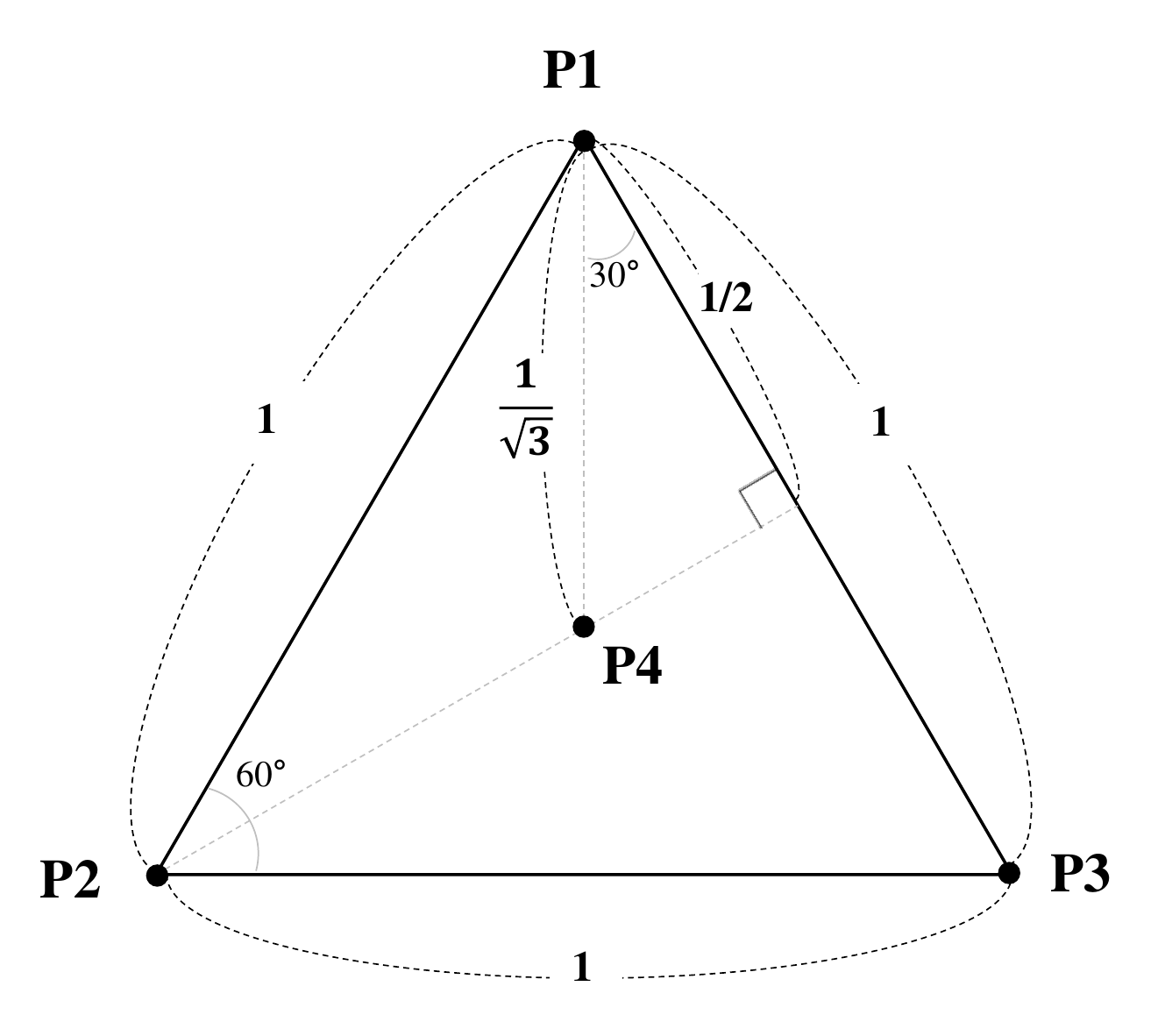}}
    \caption[Counterexample for MaxMinSum]{Counterexample for MaxMinSum.}
    \label{fig:thm-MaxMinSum}
\end{figure}

\begin{obs}
\label{obs-MaxSumMin}
The variety of a product set by \emph{MaxSumMin} is not necessarily non-decreasing when any one
product is added to the existing product set.
\end{obs}

\begin{proof}
\label{pf_obs-MaxSumMin}
To prove the theorem, we give a counterexample on $v_{I} \leq v_{I\cup\{k\}}$ for \emph{MaxSumMin},
which is $\sum_{i \in I} \min_{j \in I\setminus{i}} d_{ij}$ $\leq \sum_{i \in I\cup\{k\}}$  $\min_{j \in I\cup\{k\}\setminus{i}} d_{ij}$. We consider a product set $I$ with two products, namely P1 and P2, whose feature vectors form a line with a length of 2, as shown in Figure \ref{fig:thm-MaxSumMin}. The variety of the product set $v_{I}$ by \emph{MaxSumMin} is $2 + 2 = 5$. Suppose another product, P3, is added to the product set $I$, where P3 is located in the middle of P1 and P2 on the line. In this case, the variety of the product set $v_{I\cup\{P3\}}$ by \emph{MaxSumMin} is $1 + 1 + 1 = 3$. In this example, \emph{MaxSumMin} decreases when a product P3 is added to the existing product set $I$, which contradicts the third condition $v_{I} \leq v_{I\cup\{k\}}$.
\end{proof}

\begin{figure}
    \centerline{\includegraphics[width=6cm]{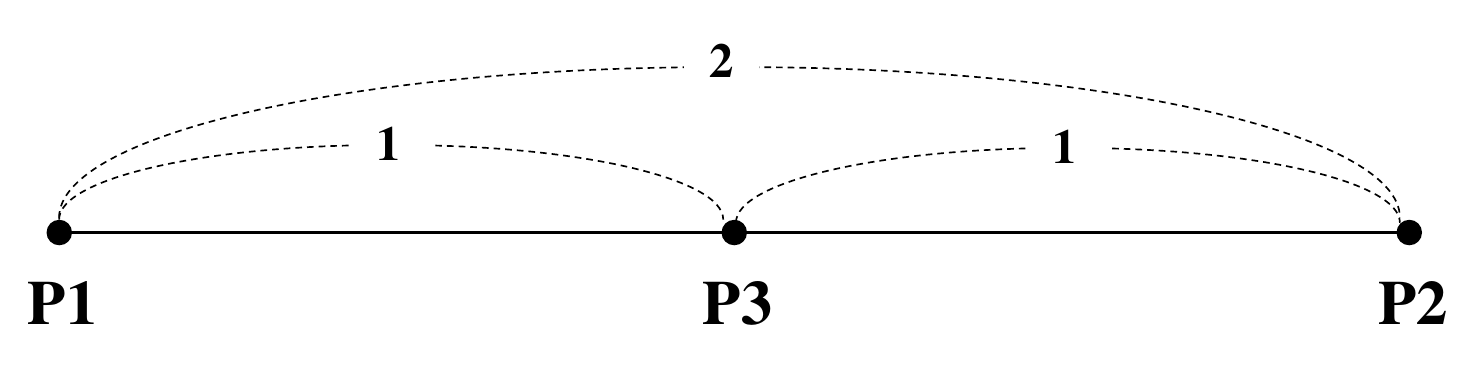}}
    \caption[Counterexample for MaxSumMin]{Counterexample for MaxSumMin.}
    \label{fig:thm-MaxSumMin}
\end{figure}

\begin{thm}
\label{thm-MaxMean}
The variety of a product set determined by \emph{MaxMean} is non-decreasing when any one product is added to
the existing product set.
\end{thm}

\begin{proof}
\label{pf_thm-MaxMean}
To prove the theorem, we need to show $v_{I} \leq v_{I\cup\{k\}}$ for \emph{MaxMean}, which is
$\sum_{i,j \in I, i<j} d_{ij} / |I| \leq \sum_{i,j \in I\cup\{k\}, i<j} d_{ij} / |I\cup\{k\}|$.
The right hand side can be rewritten as $(\sum_{i,j \in I, i<j} d_{ij}/|I| + \sum_{j \in I} d_{kj})
/ |I\cup\{k\}|$.
Let $D = \sum_{i,j \in I, i<j} d_{ij}$ and $N = |I|$.
Then, we want to show $D / N \leq (D + \sum_{j \in I} d_{kj}) / (N + 1)$.
This can be rewritten as $D / N \leq \sum_{j \in I} d_{kj}$.
By the triangle inequality, $d_{12} \leq d_{k1} + d_{k2}$, \ldots, $d_{ij} \leq d_{ki} + d_{kj}$.
By summing up these inequalities, we get $D \leq (N - 1) \sum_{j \in I} d_{kj}$.
Therefore, $D / N \leq D / (N - 1) \leq \sum_{j \in I} d_{kj}$, which shows that $v_{I} \leq
v_{I\cup\{k\}}$ for \emph{MaxMean}.
\end{proof}

Theorem \ref{thm-MaxMean} shows that \emph{MaxMean} meets the first condition (monotonicity), whereas
\emph{MaxMinSum} and \emph{MaxSumMin} do not, by Observations \ref{obs-MaxMinSum} and
\ref{obs-MaxSumMin}. Additionally, \emph{MaxSumSum} also satisfies the
monotonicity condition, although we omit the proof because it is trivial (If any one product is added to the
existing product set, the variety of new products by \emph{MaxSumSum} must increase because
all possible distances from the additional product to the existing products are added to the
existing variety).

\subsection{Linearity condition} \label{app:linearity}

\begin{figure}
    \centerline{\includegraphics[width=6cm]{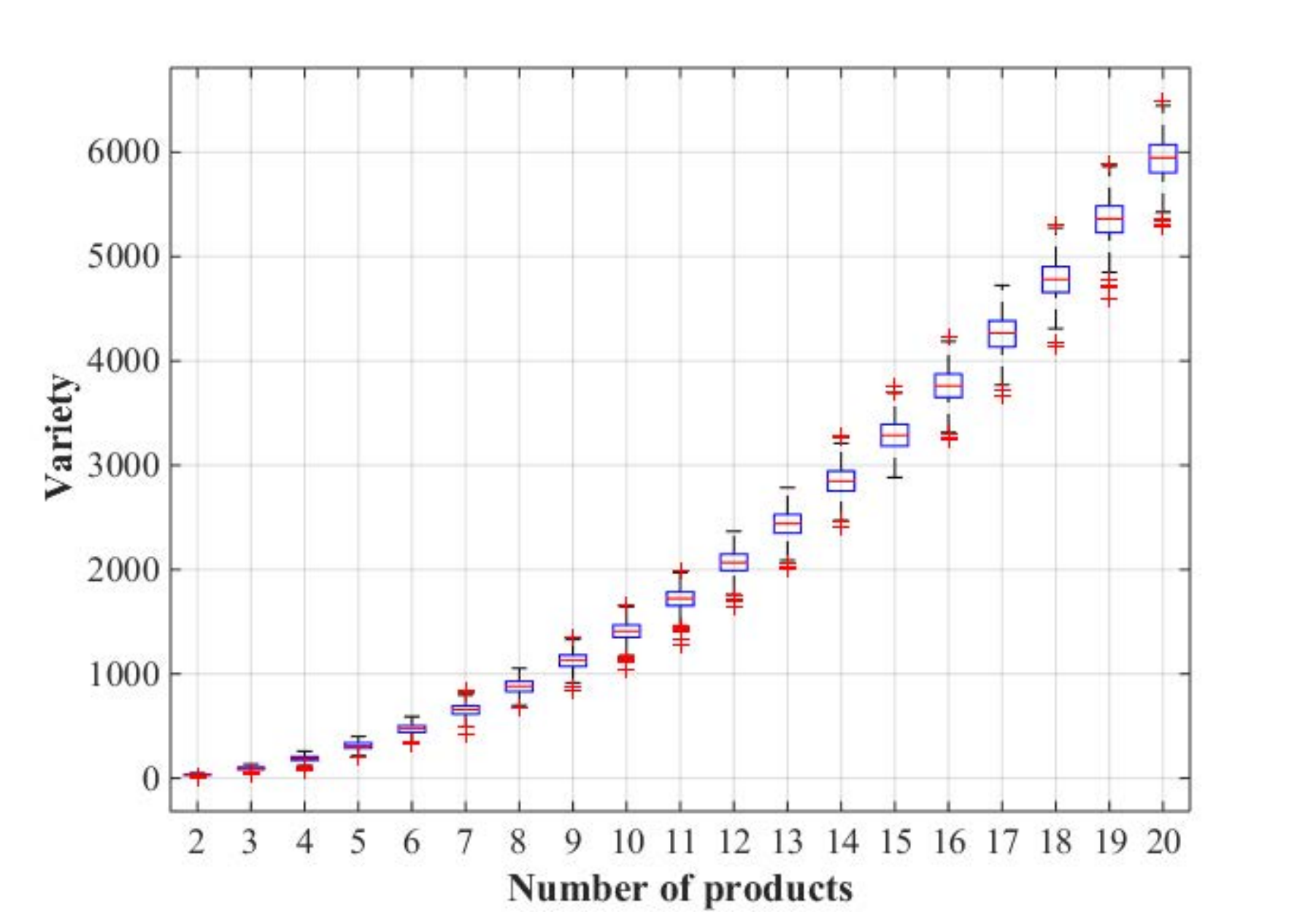}}
    \caption[Variety simulation on MaxSumSum]{Simulation results for variety by number of products -
    MaxSumSum.}
    \label{fig:simulation-MaxSumSum}
\end{figure}

\begin{figure}
    \centerline{\includegraphics[width=6cm]{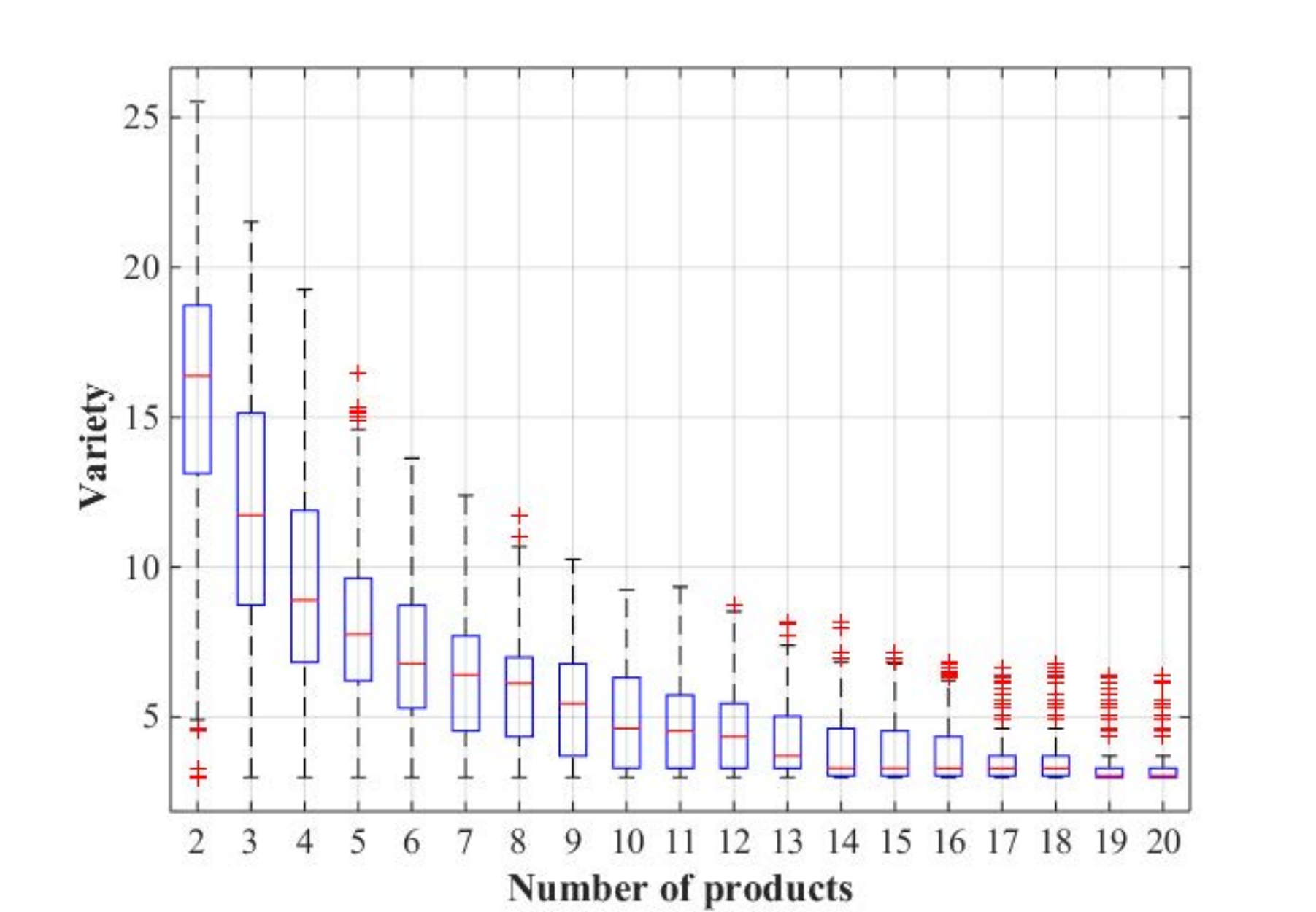}}
    \caption[Variety simulation on MaxMin]{Simulation results for variety by number of products -
    MaxMin.}
    \label{fig:simulation-MaxMin}
\end{figure}

\begin{figure}
    \centerline{\includegraphics[width=6cm]{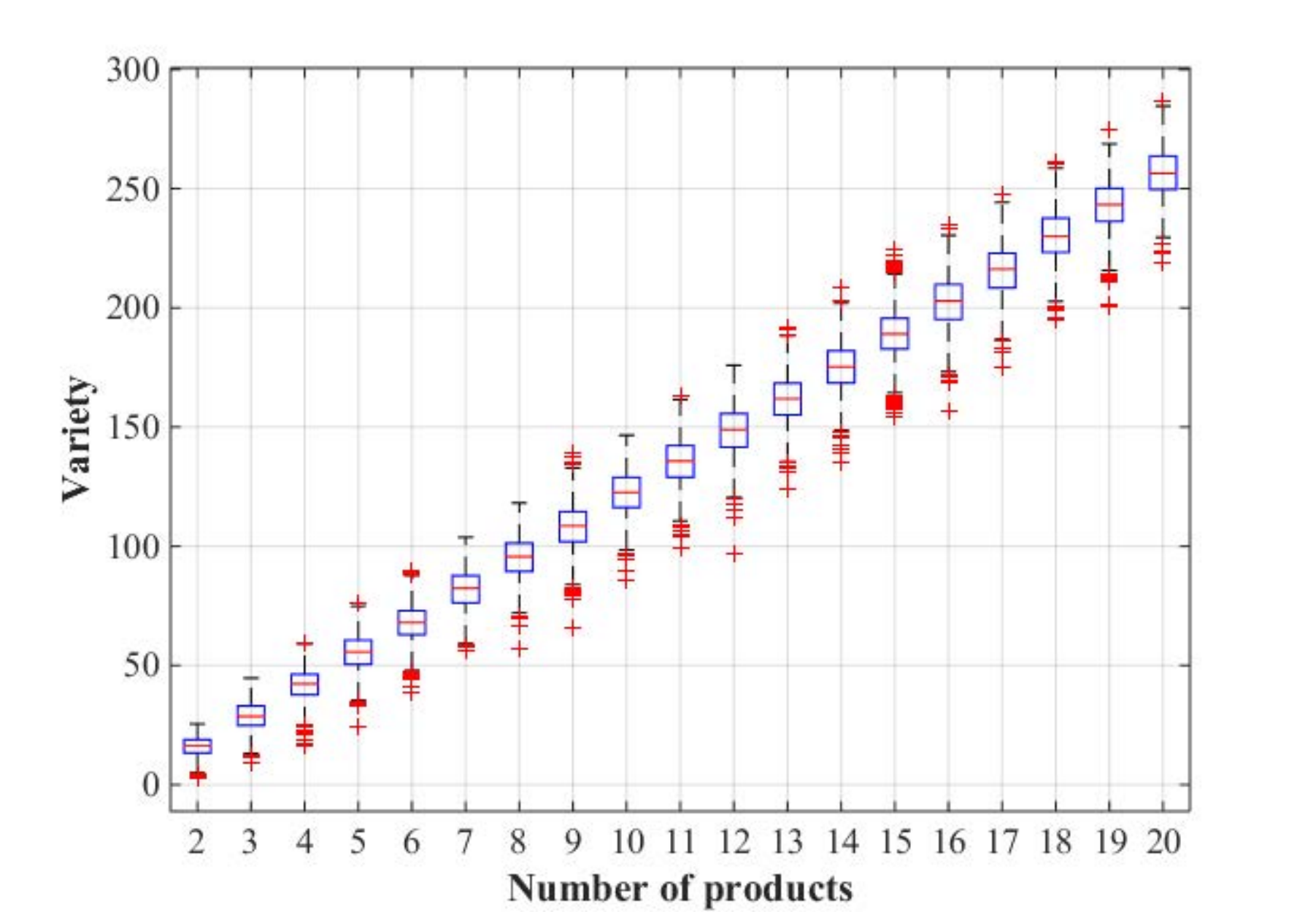}}
    \caption[Variety simulation on MaxMinSum]{Simulation results for variety by number of products -
    MaxMinSum.}
    \label{fig:simulation-MaxMinSum}
\end{figure}

\begin{figure}
    \centerline{\includegraphics[width=6cm]{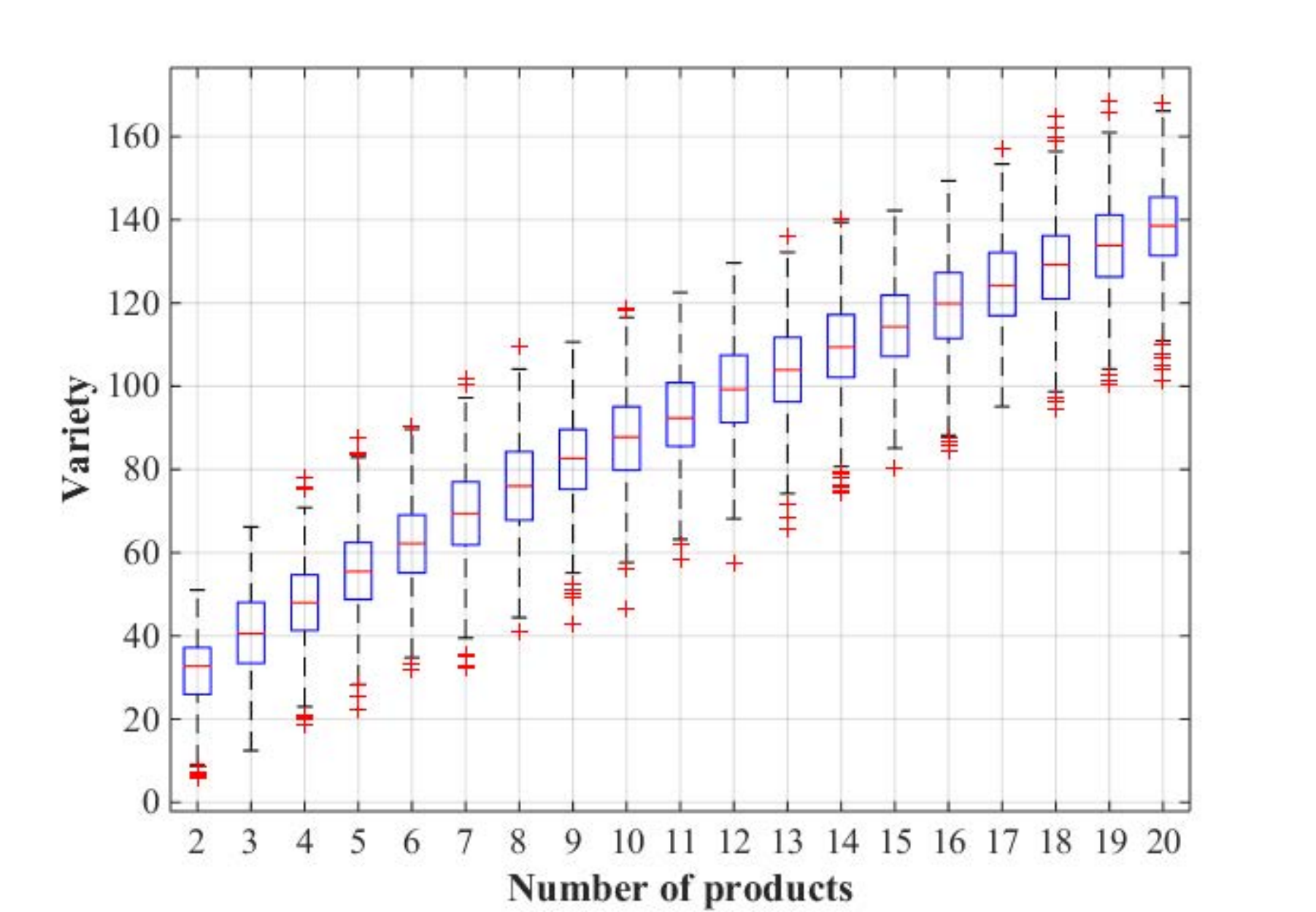}}
    \caption[Variety simulation on MaxSumMin]{Simulation results for variety by number of products -
    MaxSumMin.}
    \label{fig:simulation-MaxSumMin}
\end{figure}

\begin{figure}
    \centerline{\includegraphics[width=6cm]{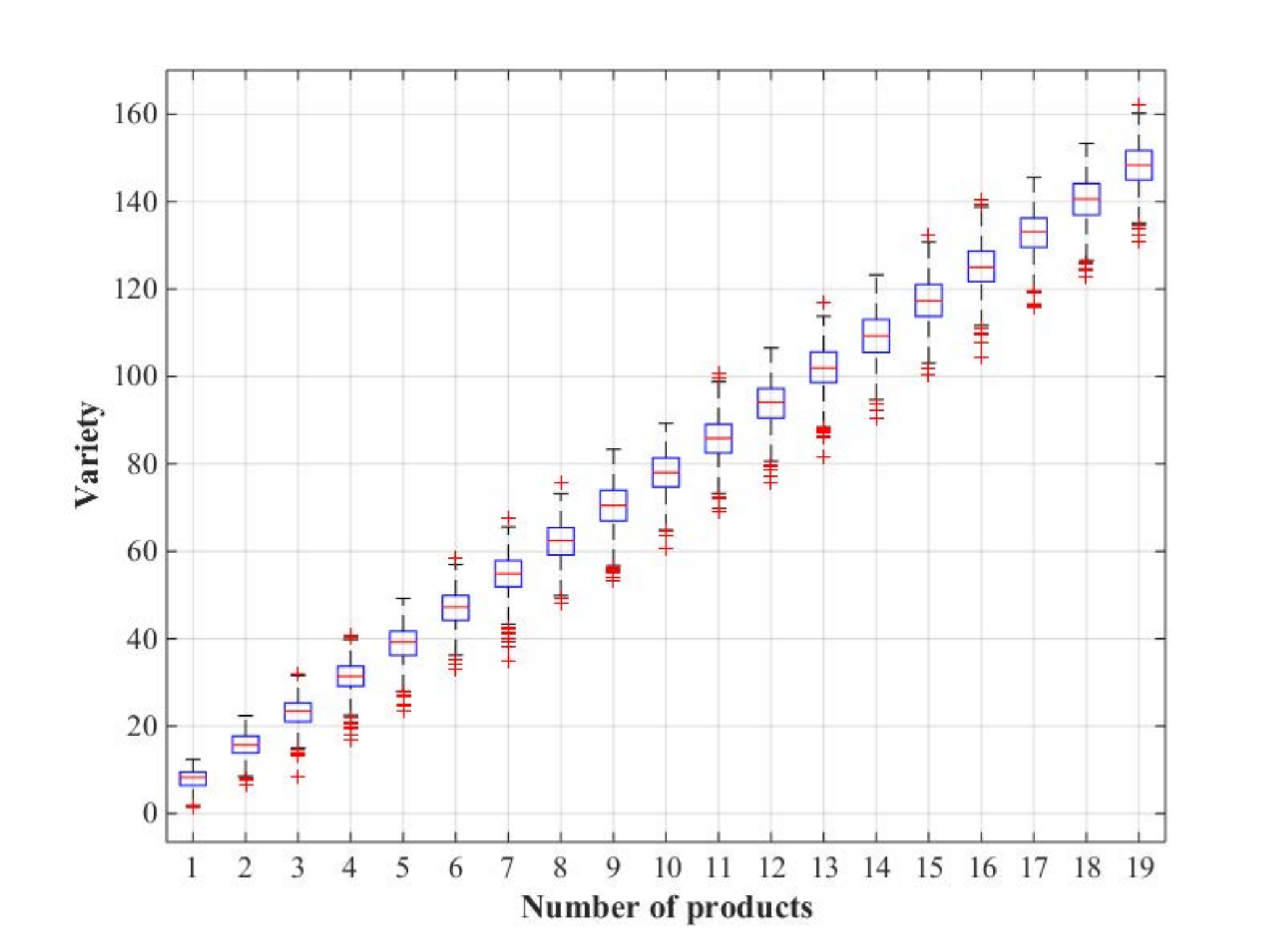}}
    \caption[Variety simulation on MaxMean]{Simulation results for variety by number of products -
    MaxMean.}
    \label{fig:simulation-MaxMean}
\end{figure}

For the second condition (linearity), we numerically investigate the behavior of the variety measures with respect
to the number of products in the product set. We use the product images of men's shirts from the 2017 S/S
of a fashion brand of the South Korean fashion retailer as the population of products. There are 35
different products and corresponding images. For each variety measure, we randomly select products
in the population with the given size of the product set and evaluate the variety of the product
set. We set the number of products in the product set to be from 2 to 20. For each number, the
product selection is repeated 1,000 times. Figures \ref{fig:simulation-MaxSumSum} to
\ref{fig:simulation-MaxMean} show the simulation results for the variety measures. In Figure
\ref{fig:simulation-MaxSumSum}, we can see that the variety measured by \emph{MaxSumSum} increases
exponentially in the number of products. This is not desirable because the measure is overly dependent
on the number of products and thus the visual differences between the products, which are important for
fashion products, may be overlooked. Figure \ref{fig:simulation-MaxMin} demonstrates that
\emph{MaxMin} is inappropriate because the variety decreases as the number of products increases,
which is counter-intuitive. In Figures \ref{fig:simulation-MaxMinSum},
\ref{fig:simulation-MaxSumMin}, and \ref{fig:simulation-MaxMean}, we can observe linear increases, proportional to the number of products, for \emph{MaxMinSum}, \emph{MaxSumMin}, and
\emph{MaxMean}, respectively. Therefore, we can state that \emph{MaxMinSum}, \emph{MaxSumMin}, and
\emph{MaxMean} satisfy the second condition.
Therefore, \emph{MaxMean} is the only variety measure that satisfies all
of the conditions required by the fashion variety model: it is non-decreasing when any
one product is added to the existing product set, and tends to be
linearly proportional to the number of distinct products in a set.

\section{Distribution optimization} \label{app:optimization}
The index sets, decision variables, and parameters of the optimization model are summarized below.

\textbf{Index sets}
\begin{itemize}
  \item [] $I$: target articles to be distributed to stores
  \item [] $S$: stores to which the articles are distributed
\end{itemize}

\textbf{Decision variables}
\begin{itemize}
  \item [] $x_{is}$: distribution quantity for article $i$ to store $s$
  \item [] $y_{is}$: binary variable: 1 if style $i$ is distributed to store $s$, or 0 otherwise
  \item [] $r_{s}$: reciprocal of the number of different styles distributed to store $s$
      $(1/\sum_i y_{is})$
  \item [] $u_{is}$: auxiliary variable to linearize the mean calculation
  \item [] $w_{ijs}$: auxiliary variable to linearize the variety calculation
  \item [] $v_{s}$: variety of distributed styles at store $s$
\end{itemize}

\textbf{Parameters}
\begin{itemize}
  \item [] $q_{s}$: ideal quantities for aggregated articles at store $s$
  \item [] $\alpha$: maximum allowable deviation of the actual distribution quantity from the ideal one
  \item [] $p_{i}$: planned total quantity for article $i$
  \item [] $m_{i}$: minimum distribution quantity for article $i$
  \item [] $M$: large number for the $y_{is}$ related constraint (set to $q_{s}$)
  \item [] $d_{ij}$: distance (dissimilarity) between article $i$ and article $j$
\end{itemize}

\begin{align}
\label{eqn:maxmean-obj} \mbox{Maximize   }& \sum_s v_{s} &&\\
\label{eqn:maxmean-c1}  \mbox{Subject to }& \sum_i x_{is} \leq (1+\alpha)q_{s} && \forall s \in S
\\
\label{eqn:maxmean-c2}                   & \sum_i x_{is} \geq (1-\alpha)q_{s} && \forall s \in S \\
\label{eqn:maxmean-c3}                   & \sum_s x_{is} \leq p_{i}           && \forall i \in I \\
\label{eqn:maxmean-c4}                   & x_{is} \geq m_{i}y_{is}            && \forall i \in I, s
\in S \\
\label{eqn:maxmean-c5}                   & x_{is} \leq M y_{is}               && \forall i \in I, s
\in S \\
\label{eqn:maxmean-c6}                   & \sum_i y_{is} \geq 2               && \forall s \in S \\
\label{eqn:maxmean-c7}                   & u_{is} \geq r_{s} + y_{is} - 1, \, u_{is} \leq r_{s}, \,
u_{is} \leq y_{is} && \forall i \in I, s \in S \\
\label{eqn:maxmean-c8}                   & \sum_i u_{is} = 1                  && \forall s \in S \\
\label{eqn:maxmean-c9}                   & w_{ijs} \geq r_{s} + y_{is} + y_{js} - 2 && \forall i,
j(>i)\in I, \, s \in S \\
\label{eqn:maxmean-c10}                  & w_{ijs} \leq y_{is}, \, w_{ijs} \leq y_{js}, \, w_{ijs} \leq r_{s} && \forall i,
j(>i)\in I, \, s \in S \\
\label{eqn:maxmean-c11}                  & v_{s} = \sum_i \sum_{j<i} d_{ij}w_{ijs}, && \forall s
\in S \\
\label{eqn:maxmean-c14}                  & w_{ijs} \geq 0      && \forall i,j(>i)\in I, 
s
\in S \\
\label{eqn:maxmean-c15}                  & u_{is} \geq 0      && \forall i \in I, s
\in S \\
\label{eqn:maxmean-c12}                  & x_{is} \geq 0, \mbox{integer}      && \forall i \in I, s
\in S \\
\label{eqn:maxmean-c13}                  & y_{is} \in \{0, 1\}                && \forall i \in I, s
\in S.
\end{align}

In the mathematical model, the objective function (\ref{eqn:maxmean-obj}) maximizes the sum of the
product varieties of the stores. Constraints (\ref{eqn:maxmean-c1}) and (\ref{eqn:maxmean-c2}) bound
the distribution quantity $x_{is}$ to be within the allowable range of $\alpha$\% deviation from
the desirable total quantity $q_{s}$. Constraint (\ref{eqn:maxmean-c3}) is the resource constraint, which
restricts the total distribution quantity of each product to its planned quantity $p_i$.
Constraints (\ref{eqn:maxmean-c4}) and (\ref{eqn:maxmean-c5}) ensure that $x_{is}$ is greater than or
equal to the minimum distribution quantity $m_i$ if product $i$ is distributed to store $s$.
Constraint (\ref{eqn:maxmean-c6}) is added because the variety measure is valid if at least two
different products are distributed. Constraints (\ref{eqn:maxmean-c7}), (\ref{eqn:maxmean-c8}),
(\ref{eqn:maxmean-c9}), (\ref{eqn:maxmean-c10}), and (\ref{eqn:maxmean-c11}) define the
\emph{MaxMean} variety. Constraints (\ref{eqn:maxmean-c14}), (\ref{eqn:maxmean-c15}), (\ref{eqn:maxmean-c12}) and (\ref{eqn:maxmean-c13}) condition
the decision variables.

\end{document}